\documentclass[a4paper,12pt]{article}
\usepackage{amsmath,amsfonts,amsthm,enumerate}
\usepackage{a4wide} 
\usepackage{graphicx} 
\usepackage{tikz}

\def\Ghat{\widehat G}
\def\hatG{\widehat G}

\def\Vhat{\widehat V}
\def\hatV{\widehat V}
\def\Ehat{\widehat E}

\def\hatx{\hat{x}}
\def\haty{\hat{y}}

\def\hatw{\widehat{w}}

\def\NP{\mathrm{NP}}
\def\RP{\mathrm{RP}}

\def\pr#1#2{\langle#1,#2\rangle}

\let\epsilon=\varepsilon  
\let\Upsilon=\varUpsilon
\let\rho=\varrho

\newtheorem{theorem}{Theorem}
\newtheorem{lemma}[theorem]{Lemma}
\newtheorem{corollary}[theorem]{Corollary}

\def\tutte{\textsc{Tutte}}
\def\planartutte{\textsc{PlanarTutte}}
\def\multitutte{\textsc{MultiTutte}}
\def\countprob#1#2#3{
\begin{list}{}{\itemsep=-2pt \leftmargin=\parindent}
 \item {\it Name.} #1
 \item {\it Instance.} #2
 \item {\it Output.} #3
\end{list}}
\def\decisionprob#1#2#3{
\begin{list}{}{\itemsep=-2pt \leftmargin=\parindent}
 \item {\it Name.} #1
 \item {\it Instance.} #2
 \item {\it Question.} #3
\end{list}}

\begin{document}
 
\title{Inapproximability of the Tutte polynomial\\ 
of a planar graph\thanks{
This paper is available on the ArXiv
at 
http://arxiv.org/abs/0907.1724.
The work was partially supported by the EPSRC grant 
\it  Computational Counting
}}
\author{Leslie Ann Goldberg\\
Department of Computer Science\\
University of Liverpool\\
Ashton Bldg, Liverpool L69 3BX, UK \and
Mark Jerrum \\
School of Mathematical Sciences\\
Queen Mary, University of London\\
London E1 4NS, UK }

\date{March 21, 2011}
\maketitle

\begin{abstract}
\noindent
The Tutte polynomial of a graph $G$ is a two-variable polynomial $T(G;x,y)$ that encodes many interesting properties 
of the graph. 
We study the complexity of the following problem, for rationals $x$ and~$y$: 
given 
as input a planar graph $G$, 
determine $T(G;x,y)$.
Vertigan completely mapped the complexity of \emph{exactly} computing the Tutte polynomial of
a planar graph. He showed that the problem can be solved in polynomial time if 
$(x,y)$ is on the hyperbola $H_q$ given by $(x-1)(y-1)=q$ for $q=1$ or $q=2$ or
if $(x,y)$ is one of the two special points  $(x,y)=(-1,-1)$ or $(x,y)=(1,1)$. Otherwise, the problem is \#P-hard.
In this paper, we consider the problem of {\it approximating} $T(G;x,y)$,
in the usual sense of ``fully polynomial randomised approximation scheme'' or FPRAS\null.
Roughly speaking, an FPRAS is required to produce, in polynomial time and with 
high probability, an answer that has 
small relative error.
Assuming that NP is different from RP, we show that there is no FPRAS for the Tutte polynomial 
in a large portion of the $(x,y)$ plane.
In particular, there is no FPRAS if $x>1,y<-1$ or if $y>1,x<-1$ or if $x<0,y<0$ and $ q>5$.
Also, there is no FPRAS if $x<1,y<1$ and $ q=3$.
For $q>5$, our result is intriguing because it shows that there is no FPRAS 
at $(x,y)=(1-q/(1+\epsilon),-\epsilon)$ for any positive $\epsilon$ but it leaves open the limit point $\epsilon=0$,
which corresponds to approximately counting $q$-colourings of a planar graph.
\end{abstract}

\section{Introduction}
 
\subsection{The Tutte Polynomial}

The Tutte polynomial of a graph $G=(V,E)$  (see \cite{Tutte84, Welsh93})
is the two-variable polynomial

\begin{equation}
\label{eq:tutte}
T(G;x,y) = \sum_{A\subseteq E} {(x-1)}^{\kappa(V,A)-\kappa(V,E)}
{(y-1)}^{|A|-n+\kappa(V,A)},
\end{equation}
where $\kappa(V, A)$ denotes the number of connected components of the graph $(V,A)$
and $n=|V|$.
Following the usual convention for the Tutte polynomial~\cite{Sokal05} a graph is allowed
to have loops and/or multiple edges.

Many interesting properties of a graph correspond to evaluations of the Tutte polynomial
at different points $(x,y)$. For example, the number of spanning trees of a
connected graph $G$ is $T(G;1,1)$, the number of acyclic orientations
is $T(G;2,0)$, and the reliability probability $R(G;p)$ of the graph
is an easily-computed multiple of $T(G;1,1/(1-p))$.
For a positive integer~$q$, the Tutte polynomial along the hyperbola $H_q$ given 
by $(x-1)(y-1)=q$ corresponds to the Partition function of the $q$-state Potts model.
See Welsh's book~\cite{Welsh93} for details.

Two particularly interesting Tutte invariants correspond to evaluations
along the $x$ axis and the $y$ axis. 
In particular, 
\begin{itemize} 
\item The chromatic polynomial
$P(G;\lambda)$ of a graph~$G$ with $n$ vertices, $m$ edges and
$k$ connected components is given
by
$$P(G;\lambda)= {(-1)}^{n-k} \lambda^{k}
T(G;1-\lambda,0).$$
When $\lambda$ is a positive integer, $P(G;\lambda)$ counts
the proper $\lambda$-colourings of~$G$.
\item The flow polynomial $F(G;\lambda)$ is given
by
$$F(G;\lambda) = {(-1)}^{m -n+k} T(G;0,1-\lambda).$$
When $\lambda$ is a positive integer, $F(G;\lambda)$
counts the nowhere-zero $\lambda$-flows of~$G$.
\end{itemize}

\subsection{Evaluating the Tutte Polynomial}

For fixed rational numbers $x$ and $y$, 
consider the following computational problem.
\countprob{$\tutte(x,y)$.}{A graph $G=(V,E)$.}{$T(G;x,y)$.}
The parameters $x$ and $y$ are fixed in advance and are not considered
part of the problem instance.   Each choice for $x$ and~$y$ 
defines a distinct computational problem.

Jaeger, Vertigan and Welsh~\cite{JVW90} have
completely mapped the complexity of $\tutte(x,y)$.
They have shown  that $\tutte(x,y)$ is
in FP for any point $(x,y)$ on the hyperbola $H_1$
and when $(x,y)$
is one of the special points
$(1,1)$, $(0,-1)$, $(-1,0)$, and $(-1,-1)$.
They showed that $\tutte(x,y)$ is \#P-hard for every other
pair of rationals $(x,y)$.
See \cite{papadim94} for definitions of FP and \#P;
informally, FP is the extension of the class P from predicates
to more general functions, and \#P is the counting analogue of~NP\null.
Jaeger et al.\ also investigated
the complexity of
evaluating the Tutte polynomial when $x$ and~$y$
are real or complex numbers, but that is beyond the scope of this paper.

Vertigan \cite{vertigan} considered the restriction of $\tutte(x,y)$ in which the
input is restricted to be a \emph{planar} graph.
\countprob{$\planartutte(x,y)$.}{A planar graph $G=(V,E)$.}{$T(G;x,y)$.}
He showed that $\planartutte(x,y)$ is in FP for any point $(x,y)$ on the
hyperbolas $H_1$ or $H_2$,
and when $(x,y)$ is one of the special points $(1,1)$ and $(-1,-1)$.
He showed that $\planartutte(x,y)$  is \#P-hard for every other
pair of rationals $(x,y)$.
The hyperbola~$H_2$ is of particular interest, as the Tutte polynomial
here corresponds to the partition function of the celebrated Ising model
in statistical physics.

\subsection{Approximating the Tutte polynomial}

A \emph{fully polynomial randomised approximation scheme} (FPRAS)
for $\tutte(x,y)$
is a randomised algorithm 
that takes as input a graph $G$ and a constant $\epsilon\in(0,1)$
and outputs a value $Y$ such that, with probability
at least $3/4$, 
$e^{-\epsilon}\,T(G;x,y)\leq Y\leq e^{\epsilon}\,T(G;x,y) $.
The running time of the algorithm is bounded from above by
a polynomial in $n$ (the number of vertices of $G$)  and $\epsilon^{-1}$.
An FPRAS for $\planartutte(x,y)$ is defined 
similarly.  See \cite{jerrum} for further details on fully polynomial randomised approximation schemes.

\begin{figure}[t]
\centering
\includegraphics[width=9truecm]{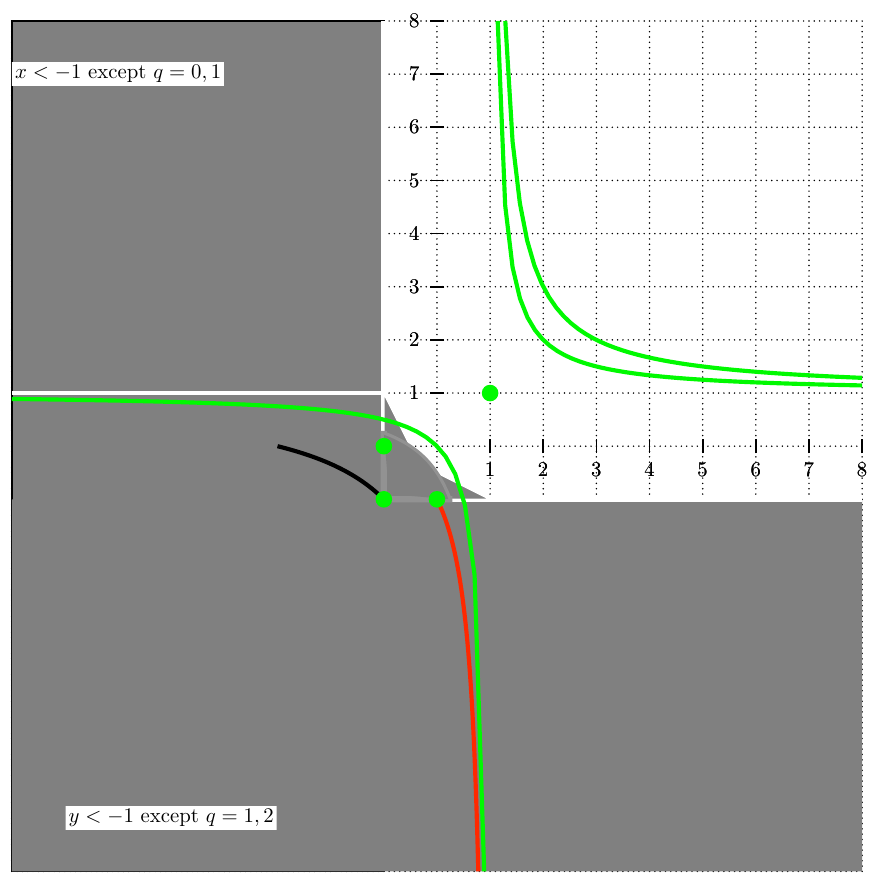}
\caption{The result from \cite{tuttepaper}.
Green points are FPRASable, red points are equivalent to counting 
perfect matchings
and gray points are not FPRASable unless $\mathrm{RP}=\mathrm{NP}$.
We don't know about white points.
The black line, which is the portion of the hyperbola $q=4$
lying in $y\in(-1,0)$, is \#P-hard.  The black points
are harder than gray in a complexity-theoretic sense.  
The black region is presumably more extensive than shown.}
\label{fig:one}
\end{figure}

In earlier work~\cite{tuttepaper}, we considered the problem of determining
for which points $(x,y)$ there is an FPRAS for $\tutte(x,y)$.
Our results are summarized in Figure~\ref{fig:one}.
In particular, under the assumption $\RP\neq \NP$, we showed the following.
\begin{enumerate}[(1)]
\item If $x<-1$ and $(x,y)$ is not on $H_0$ or $H_1$, then
there is no FPRAS at $(x,y)$.
\item If $y<-1$ and $(x,y)$ is not on $H_1$ or $H_2$,
then there is no FPRAS at $(x,y)$.
\item There is no FPRAS at points $(x,y)$ lying in certain regions 
in the vicinity of the origin, contained in the square $-1<x,y<1$.
\item If $(x,y)$ is on $H_2$ and $y<-1$ then approximating
$T(G;x,y)$ is equivalent in difficulty to approximately counting
perfect matchings (resolving the complexity of this is a well-known and interesting open problem).
\end{enumerate}

An interesting consequence of these results is that,
under the assumption $\RP\neq\NP$, there is no
FPRAS at the point $(x,y)=(0,1-\lambda)$ when $\lambda>2$ is a positive
integer. Thus, there is no FPRAS for counting nowhere-zero $\lambda$ flows
for $\lambda>2$. This is   interesting  
since the corresponding decision
problem is in~P, for example, for $\lambda=6$. See \cite{tuttepaper} for details.

\subsection{Approximating the Tutte polynomial of a planar graph}

In this paper we consider the problem of determining for
which points $(x,y)$ there is an FPRAS for $\planartutte(x,y)$.
The results of \cite{tuttepaper} do not help us here because all of
the constructions are badly non-planar.
Our results are summarised in Figure~\ref{fig:two}.

\begin{figure}[t]
\centering
\includegraphics[width=9truecm]{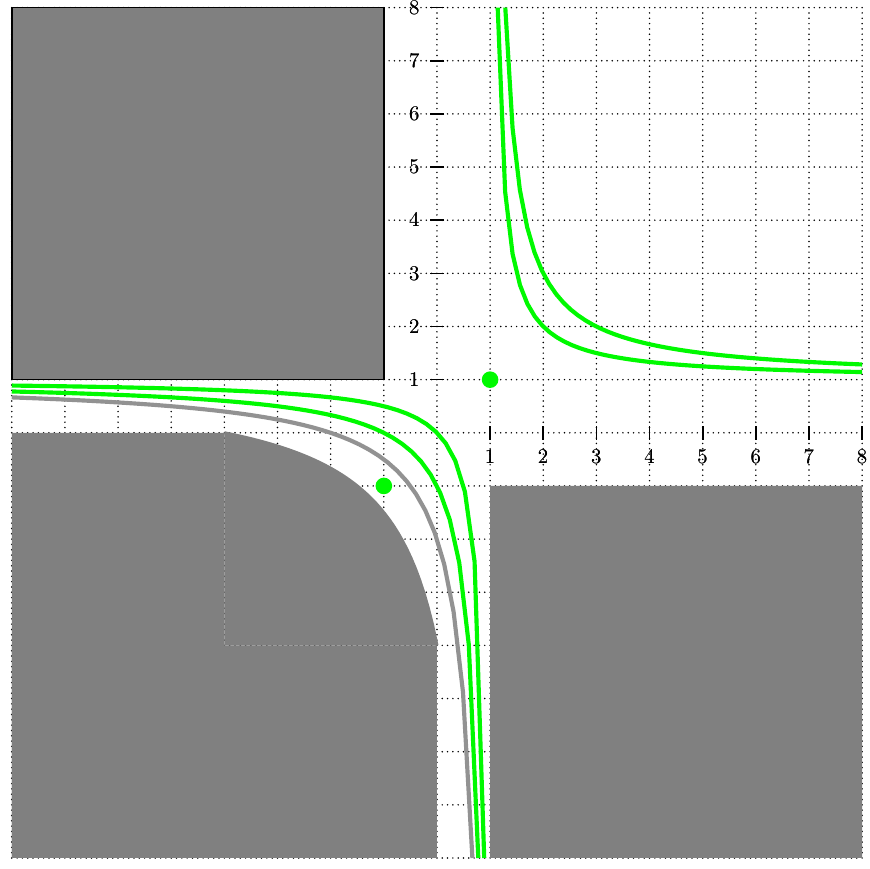}
\caption{The complexity of the planar case. 
The shaded gray regions are shown to be intractable in 
Corollary~\ref{cor:x<-1}.
The lower branch of the $q=3$ hyperbola (also depicted in gray)
is shown to be intractable in Lemma~\ref{lem:q=3}.
As Vertigan has shown~\cite{vertigan}, it is easy to compute the polynomial exactly on the hyperbolas 
$q=1$ and $q=2$ and at the 
two special points $(1,1)$ and $(-1,-1)$. (These are shown in green.)} 
\label{fig:two}
\end{figure}

In particular, under the assumption $\RP\neq \NP$, Corollary~\ref{cor:x<-1}
and Lemma~\ref{lem:q=3}
show that there is no FPRAS for $\planartutte(x,y)$ in the following
cases:
\begin{enumerate}
\item $x<0$, $y<0$ and $q>5$;
\item $x<1$, $y<1$ and $q=3$;
\item $x>1$, $y<-1$;
\item $y>1$, $x<-1$.
\end{enumerate}
 
For  integer $q\geq 4$, the point $x=1-q$, $y=0$ is
of particular interest. As noted earlier $T(G;x,y)$ 
gives the number of proper $q$-colourings of~$G$.
By the $4$-colour theorem, there is at least one $q$-colouring,
so the corresponding decision problem is trivial, but it
is not clear whether there is an FPRAS.
For $q\geq 5$, our result shows that there is no FPRAS for 
any nearby point
$x=1-q/(1+\epsilon)$, $y=-\epsilon$ on 
the hyperbola $H_q$ (for any $\epsilon>0$).
However, the case of colourings itself (corresonding
to the limit point $\epsilon=0$) remains open.
The same intriguing situation occurs
with the flow polynomial points $x=0$, $y=1-q$.

In a recent posting on ArXiv, Kuperberg \cite{Kuperberg} 
independently offers a proof
sketch, based on the complexity theory of quantum computation,
of a result closely related to ours.  If the details in the proof
sketch can be filled in, then it will strengthen our result in the
negative quadrant by (i) relaxing the condition $q \geq 5$ to $q \geq 4$,
and (ii) strengthening the conclusion to \#P-hardness.

\subsection{The multivariate formulation of the Tutte polynomial}

As in \cite{tuttepaper}, we need the multivariate
formulation of the Tutte polynomial in order to 
prove our results. The multivariate formulation is also
known as the \emph{random cluster model}~\cite{Welsh93, Sokal05}.
For $q\in\mathbb{Q}$ and a graph $G=(V,E)$ with edge weights
$w:E\to\mathbb{Q}$,
  the multivariate Tutte polynomial of $G$ is
defined by
$Z(G;q,w)=\sum_{A\subseteq E}
w(A)
q^{\kappa(V,A)}$, where $w(A) = \prod_{e\in A}w(e)$.

Suppose $(x,y)\in\mathbb{Q}^2$ and $q=(x-1)(y-1)$.
For a graph $G=(V,E)$, let $w:E\to\mathbb{Q}$ be the constant
function which maps every edge to the value $y-1$.
Then (see, for example \cite[(2.26)]{Sokal05})
\begin{equation}
T(G;x,y) = {(y-1)}^{-n}{(x-1)}^{-\kappa(E)} Z(G;q,w).
\label{eq:rcequiv}
\end{equation}

So approximating $T(G;x,y)$ is equivalent in
difficulty to approximating $Z(G;q,w)$ for the
constant function $w(e)=y-1$.
However, the multivariate formulation is more general, because
we can assign different weights to different edges of $G$.

Consider the following computational problem, which is a planar version of
one that we considered in~\cite{tuttepaper}.
\countprob{$\multitutte(q;\alpha_1,\alpha_2,\alpha_3)$.}%
{A planar graph $G=(V,E)$ with
edge labelling $w:E\to\{\alpha_{1},\alpha_2,\alpha_3\}$.}%
{$Z(G;q,w)$.}
Our main tool in proving inapproximability (Lemma~\ref{lem:bedrock} below)
is showing that $$\multitutte(q;\alpha_1,\alpha_2,\alpha_3)$$ 
is difficult to approximate 
if $\alpha_1\notin [-2,0]$, $\alpha_2\in(-2,0)$ and $\alpha_3<-1$.
(Note that $\alpha_3$ might be equal to $\alpha_1$ or $\alpha_2$.)

\section{Technical Preparation}
In this section we introduce a gadget (weighted graph) with certain 
useful properties.  Although the graph is very simple, the edge weights 
must be carefully tuned  to achieve the desired properties.  
We need to be able to ``implement'' these particular edge weights in terms of the
actual weights that are available to us. 

\subsection{Implementing new edge weights}\label{sec:weightImplement}
 Let 
$W$ be a set of edge weights 
(for example, $W$ might contain the edge weights $\alpha_1$, $\alpha_2$ and~$\alpha_3$ from above)
and fix a value $q$.
Let $w^*$ be a weight (which may not be in $W$) which we want to ``implement''.
Suppose that 
there is a planar graph~$\Upsilon$,
with distinguished vertices $s$ and~$t$
on the outer face,
and a weight function $\hatw:E(\Upsilon)\rightarrow W$
such that
\begin{equation}
\label{eq:implement}
w^* = q Z_{st}(\Upsilon)/Z_{s|t}(\Upsilon),
\end{equation}
where $Z_{st}(\Upsilon)$ denotes the contribution to
$Z(\Upsilon;q,\hatw)$ arising from edge-sets $A$ in which $s$ and $t$ are
in the same component.
That is, $Z_{st}(\Upsilon) = \sum_{A}
\hatw(A)
q^{\kappa(V,A)}$, where the sum is over subsets $A\subseteq E(\Upsilon)$ in which 
$s$ and $t$ are in the same component.
Similarly, $Z_{s|t}$ denotes the contribution to 
$Z(\Upsilon;q,\hatw)$ arising from edge-sets $A$ in which $s$ and $t$ are in different components.
In this case, we say that $\Upsilon$ and $\hatw$ implement $w^*$
(or even that $W$ implements~$w^*$).

The purpose of ``implementing''  edge weights is this.
Let $G$ be a graph with edge-weight function $w$.
Let $f$ be some edge of $G$ with edge weight $w(f)=w^*$.
Suppose that $W$ implements $w^*$.
Let $\Upsilon$ be a planar graph with distinguished vertices $s$ and $t$
with a weight function $\hatw$ satisfying (\ref{eq:implement}). 
Construct the weighted graph $G'$ 
by replacing edge $f$ with a copy of $\Upsilon$ (identify $s$ with either endpoint of $f$
(it doesn't matter which one) and identify $t$ with the other endpoint of $f$ and remove edge $f$).
Let the weight function $w'$ of $G'$ inherit weights from $w$ and $\hatw$ (so $w'(e)=\hatw(e)$ if $e\in E(\Upsilon)$ and
$w'(e)=w(e)$ otherwise).
Then the definition of the multivariate Tutte polynomial gives
\begin{equation}
\label{eq:shift}
Z(G';q,w') = \frac{Z_{s|t}(\Upsilon)}{q^2} Z(G;q,w).\end{equation}
So, as long as $q\neq 0$ and $Z_{s|t}(\Upsilon)$ is easy to evaluate,
evaluating the multivariate Tutte polynomial of $G'$ with weight function $w'$ is
essentially the same as evaluating the multivariate Tutte polynomial of $G$ with weight function~$w$.

Two especially useful implementations are series and parallel compositions.
These are explained in detail in \cite[Section 2.3]{JacksonSokal}.
So we will be brief here.
Parallel composition is the case in which $\Upsilon$ consists of two parallel edges $e_1$ and $e_2$
with endpoints $s$ and $t$ and  $\hatw(e_1)=w_1$ and $\hatw(e_2)=w_2$.
It is easily checked from Equation~(\ref{eq:implement})
that $w^* = (1+w_1)(1+w_2)-1$. Also, the extra factor in Equation~(\ref{eq:shift}) cancels,
so in this case $Z(G';q,w') = Z(G;q,w)$.

Series composition is the case in which $\Upsilon$ is a length-2 path from $s$ to $t$ consisting of edges $e_1$ and $e_2$
with $\hatw(e_1)=w_1$ and $\hatw(e_2)=w_2$.
It is easily checked from Equation~(\ref{eq:implement})
that $w^* =  w_1w_2/(q+w_1+w_2)$. Also, the extra factor in Equation~(\ref{eq:shift}) is $q+w_1+w_2$,
so in this case $Z(G';q,w') = (q+w_1+w_2) Z(G;q,w)$.
It is helpful to note that
$w^*$ satisfies
$$\left(1+\frac{q}{w^*}\right) = \left(1+\frac{q}{w_1}\right) \left(1+\frac{q}{w_2}\right).$$

We say that there is a ``shift'' 
from $(q,\alpha)$ to $(q,\alpha')$ if 
there is an implementation of $\alpha'$ consisting of some $\Upsilon$ and  
$\hatw:E(\Upsilon)\rightarrow W$ where $W$ is the singleton set $W=\{\alpha\}$.
This is the same notion of ``shift'' that we used in \cite{tuttepaper}. 
Taking $y=\alpha+1$ and $y'=\alpha'+1$
and defining $x$ and $x'$ by $q=(x-1)(y-1)=(x'-1)(y'-1)$ we
equivalently refer to this as a shift from $(x,y)$ to $(x',y')$.

Thus, the $k$-thickening of Jaeger, Vertigan and Welsh~\cite{JVW90}
is the parallel composition of $k$ edges of weight $\alpha$.
It implements $\alpha'=(1+\alpha)^k-1$ and is a shift from $(x,y)$ to
$(x',y')$ where $y' = y^k$ (and $x'$ is given by $(x'-1)(y'-1)=q$).
 Similarly, the $k$-stretch is the series composition of $k$ edges of weight $\alpha$.
It implements an $\alpha'$
satisfying
$$1+\frac{q}{\alpha'}= {\left(1+\frac{q}{\alpha}\right)}^k,$$
It is a shift from $(x,y)$ to 
$(x',y')$ where $x'=x^k$.
(In the classical bivariate $(x,y)$ parameterisation, there is effectively
one edge weight, so the stretching or thickening is applied uniformly
to every edge of the graph.)

Since it is useful to switch freely between $(q,\alpha)$ coordinates and $(x,y)$ coordinates we also
refer to the implementation in Equation~(\ref{eq:implement}) as an implementation of the 
point $(x,y)=(q/w^*+1,w^*+1)$ using the points
$$\{
(x,y)=(q/w+1,w+1)\mid w\in W
\}.$$

\subsection{Global Constants}
\label{sec:A3}

Our proofs will use several global constants which depend upon~$q$
but do not depend upon the problem instances in our reductions.
The definitions of these constants are provided here for easy reference.
The purpose of all of these constants will become clear later,  
but as a rough guide, the constants $A^-$, $A^+$, $B^-$ and $B^+$
will be lower and upper bounds on the (absolute values of the) edge weights, 
$a$ and~$b$, that we use in our gadgets.
The edge weights themselves will depend on the problem instance, but
it is important for the proof that these lower and upper bounds do not depend upon the problem instance --- they only
depend upon~$q$.

Let  $f(x)$ be the function 
$f(x) = x^3 + 3 x^2$. We start by defining  several quantities
for which the definitions differ  depending on whether $q<0$ or $q>5$.
 
\medskip\noindent{\bf Case 1: $q>5$:}\quad
$\chi = \min(1,(q-5)/6)$,
$\eta=3/4$,
$A^-=1/2$, $A^+=q$,
$B^-=q$ and 
$B^+ = 10 q^3$.
 
\medskip\noindent{\bf Case 2: $q<0$:}\quad
$\chi=\min(1,|q|)$. 
To define the other constants, it helps to make a few observations.
Let $g(y) = f(-3-y)$ and note that $g(0)=0$ and that $g'(y)<0$ for $y>0$ so $g(y)$ decreases
as $y$ increases from $0$.
Now  
let $\eta>0$ be the real solution of
$g(\eta)=q/2$.
Let $A^-=3+\eta$.
Then let $y^*>0$ be the real solution of
$g(y^*)=q$.
Let $A^+=3+y^*$.
Let $B^- = |q|/3$ and let $B^+ = 4|q|/3 + 2$.
\medskip
 
Note that, in both cases, $0< A^- <A^+$ and $0<B^-<B^+$ and $\eta>0$.
Finally,  define
\begin{itemize}
\item
$A^* = 1 + 3({A^+})^4 + 9 ({A^+})^3 + 3 ({A^+})^2 + 3 A^+ (1+|q|)$,
\item $Q = \max(|q|,|q|^{-1})$, 
\item $\mu =  q^2 A^+ ({B^+})^2$,
\item $\tau = |q|({A^+})^2 (A^++3) ({B^+})^3$, and
\item $M=\max(1,\mu,\tau)$.
\end{itemize}

\subsection{Implementing useful edge weights}
\label{sec:A4}

In much of the technical part of the paper, we will have at our disposal three edge weights,
$\alpha_1$, $\alpha_2$ and $\alpha_3$  
such that 
$\alpha_1\notin [-2,0]$, $\alpha_2\in(-2,0)$ and $\alpha_3<-1$.
($\alpha_3$ might be equal to $\alpha_1$ or $\alpha_2$.)
Now that we have defined the global constants in Section~\ref{sec:A3},
we state some lemmas 
showing that we can use $\alpha_1$, $\alpha_2$ and $\alpha_3$ to
implement certain edge weights, $a$, $b$ and~$\beta$, which we
will later use in our gadgets. 
As will become apparent below, 
the precise definitions of $a$, $b$ and $\beta$ will depend upon 
two accuracy parameters~$\rho$ and~$\hat\rho$.
When we use the lemmas, we will take these to be very small (depending on the input sizes in our reductions).
We defer the proofs of the lemmas until Section~\ref{sub:deferred} because are they mainly technical, and
are not necessary for understanding our main 
argument.

\begin{lemma}  \label{lem:impa}
Suppose $q\notin [0,5]$
 and that
$\alpha_1\notin [-2,0]$, $\alpha_2\in(-2,0)$ and $\alpha_3<-1$.
Given a positive 
constant~$\rho$
which is  sufficiently small with respect to~$q$, $\alpha_1$, $\alpha_2$, and
$\alpha_3$, 
there is a planar graph $\Upsilon$ (depending on  
$\rho$) and a weight 
function $\hatw:E(\Upsilon)\rightarrow \{\alpha_1,\alpha_2,\alpha_3\}$ 
that implements a weight 
$a$, such that
\begin{gather}
\label{eqa1} A^- \leq |a| \leq A^+,\\
\label{eqa2}   q + \rho  < f(a) \leq q + 2 \rho,\>\mbox{and }\\
\label{eqa3} |f(a)|\geq  \eta.
\end{gather}
The size of $\Upsilon$ is at most a polynomial in 
$\log(\rho^{-1})$.
\end{lemma}

\begin{lemma}  \label{lem:impb}
Suppose $q\notin [0,5]$
 and that
$\alpha_1\notin [-2,0]$, $\alpha_2\in(-2,0)$ and $\alpha_3<-1$.
Suppose, for a positive value~$\rho$,
which is  sufficiently small with respect to~$q$, $\alpha_1$, $\alpha_2$, 
and $\alpha_3$,
the value~$a$ satisfies inequalities (\ref{eqa1}), (\ref{eqa2}) and~(\ref{eqa3}).
Let 
\begin{equation}
\label{eqc} c = a^2+3a+q\end{equation}
Given a positive 
constant~$\hat\rho$
which is sufficiently small with respect to~$q$, $\alpha_1$, $\alpha_2$, and $\alpha_3$,
there is a planar graph $\Upsilon$ (depending on  
$\hat \rho$) and a weight 
function $\hatw:E(\Upsilon)\rightarrow \{\alpha_1,\alpha_2,\alpha_3\}$ 
that implements a weight 
$b$, such that
\begin{align}
\label{eqb1} & B^- \leq |b| \leq B^+,\> \mbox{and }\\
\label{eqb2} & - \hat \rho \leq b + c \leq \hat \rho.
\end{align}
The size of $\Upsilon$ is at most a polynomial in 
$\log({\hat \rho}^{-1})$.
\end{lemma}

\begin{lemma}  \label{lem:impbeta}
Suppose $q\notin [0,5]$
and that $\alpha_2\in(-2,0)$.
Given a positive 
constant~$\rho$
which is  sufficiently small with respect to~$q$ and $\alpha_2$, there is a planar graph $\Upsilon$ (depending on  
$\rho$) and a weight 
function $\hatw:E(\Upsilon)\rightarrow \{ \alpha_2 \}$ 
that implements a weight 
$\beta$, such that 
$|1+\beta|\leq \rho$.
The size of $\Upsilon$ is at most a polynomial in 
$\log(\rho^{-1})$.
\end{lemma} 
 
 \subsection{A useful gadget}
 \label{sec:A5}
 
Suppose $a$ and $b$ are edge weights.
Let $Y$ be a weighted graph with weight function $w$ defined as follows.
$Y$ will have
vertex set  $V(Y)=\{0,1,2,\overline{0},\overline{1},\overline{2}\}$. 
The edge set  $E(Y)$ of $Y$
consists of three edges
$(0,\overline{0})$, $(1,\overline{1})$ and $(2,\overline{2})$ of weight~$b$
and three edges $(\overline{0},\overline{1})$, $(\overline{1},\overline{2})$ and
$(\overline{2},\overline{0})$ of weight~$a$. 
See Figure~\ref{fig:Y}.

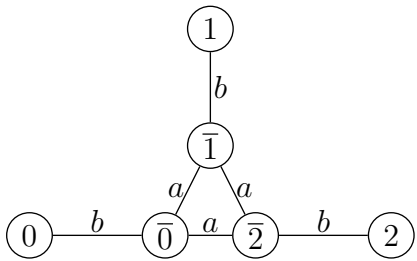
\begin{figure}

\centering{
 \begin{tikzpicture}[fill=white,scale=0.7,
line width=0.5pt,inner sep=1pt,minimum size=2.5mm]
\pgfsetxvec{\pgfpoint{1.7cm}{0cm}}
\pgfsetyvec{\pgfpoint{0cm}{1.7cm}}
\path
(3.5,2) node
[draw,fill,circle,minimum size=0.6cm](u0){ $ 0$}
(5,2) node [draw,fill,circle,minimum size=0.6cm](d){$\overline{0}$}
(6,2) node [draw,fill,circle,minimum size=0.6cm](dd){$\overline{2}$}
(5.5,3) node [draw,fill,circle,minimum size=0.6cm](e){$\overline{1}$}
(5.5,4.3) node [draw,fill,circle,minimum size=0.6cm](uv1){ $ 1$}
(7.5,2) node [draw,fill,circle,minimum size=0.6cm](uv2){ $ 2$}
    ;
\draw[-]  (u0)-- node[above] {$b$}
(d)-- node[above] {$a$} (dd)-- node[above] {$b$} (uv2);
\draw[-] (d)-- node[left] {$a$}(e)-- node[right] {$a$}(dd);
\draw[-] (e)-- node[right] {$b$} (uv1);
\end{tikzpicture}

 }
 \caption{The  gadget $Y$.}
\label{fig:Y}
\end{figure}

For a fixed $q$, let 
$Z_{0|1|2}$ denote the
contribution to $Z(Y;q,w)$ arising from edge sets $A$ in which the vertices~$0$, $1$ and $2$
are in distinct components.
Thus, $Z_{0\mid1\mid 2} = \sum_{A}
\prod_{e\in A}w(e)
q^{\kappa(V(Y),A)}$, where the sum is over all subsets $A\subseteq E(Y)$  
such that $0$, $1$ and $2$ are all in distinct components.
Similarly, let $Z_{0|12}$ denote the contribution to 
$Z(Y;q,w)$ arising from edge sets $A$ in which the vertex~$0$ is in one component
and the vertices $1$ and $2$ are in another, distinct, component.
Finally, let $Z_{012}$ denote the contribution to $Z(Y;q,w)$ denote the contribution to 
$Z(Y;q,w)$ arising from edge sets $A$ in which the vertices~$0$, $1$ and $2$ are all in the same component.
Define $c$ via Equation (\ref{eqc}). 
From the definition of $Z_{0|12}$ we see that
\begin{equation}
\label{eq:10Feb1}
Z_{0|12}=q^{2}ab^{2}(c+b) 
\end{equation}

Similarly,
$$
Z_{0|1|2}=q^{3}\big(b^{3}+3b^{2}(2a+q)
   +(3b+q)(a^{3}+3a^{2}+3aq+q^{2})\big).
$$

Let
\begin{equation}
d=a^2+3a+q+b\label{eq:d}
\end{equation}
and
\begin{equation}
 e=a^3+3a^2-q.\label{eq:e}
\end{equation}
Then 
 \begin{align} \nonumber
Z_{0|1|2} & =  -q^{3}a^{2}(a+3)(a^{3}+3a^{2}-q)
+ 
 d^3 q^3 +
 d^2 (-3a - 3a^2) q^3 + d q^3 (9 a^3 + 3 a^4 - 3 a q)
\\ \nonumber
& = 
-q^{3}a^{2}(a+3) e
+ 
 d^3 q^3 -
 d^2 (3a + 3a^2) q^3 + d q^3 (9 a^3 + 3 a^4 - 3 a q)
\\ &= \label{eq:10Feb2}
-q^{3}a^{2}(a+3) 
\left(
e
- \frac{d^3}{a^{2}(a+3)}
+  \frac{
 d^2 (3a + 3a^2)}{a^{2}(a+3)}
- \frac{  d  (9 a^3 + 3 a^4 - 3 a q)}
{a^{2}(a+3)}
\right).
\end{align}

Also 
\begin{equation}
\label{eq:10Feb3}
Z_{012} =qa^{2}(a+3)b^{3}.
\end{equation}

\subsection{A lower bound on the Tutte polynomial}

We conclude our technical preparations by presenting a lemma which 
gives a positive lower bound on the (multivariate) Tutte polynomial of a planar graph for $q>5$. The lemma is 
essentially due to Woodall~\cite[Theorem 1]{Woodall}. 
(The method can be traced 
back to~\cite{BL46}.) However, we need two slight generalisations. First,
Woodall's proof was actually about the chromatic polynomial, which corresponds to the specialisation of the Tutte polynomial
in which $w(e)=-1$ for every edge~$e$. In our lemma, we will ensure that $w(e)$ is always close to~$-1$ but it will not
be exactly equal to~$-1$. Second, Woodall's objective was to show that the polynomial is
positive. We will need something slightly stronger --- namely, a strictly positive lower bound.
Woodall's proof technique suffices to provide this.

\begin{lemma}  
\label{lem:Woodall}
Suppose $q>5$.
Suppose $\rho\in (0,1)$ and $\zeta\in (0,1)$  satisfy
$q\geq5(1+\rho)+\zeta$.
For any simple planar graph~$G=(V,E)$ and any edge-weight function $w$
satisfying
$  |1+w(e)|\leq \rho$ for all $e\in E$,
$Z(G;q,w)\geq \zeta^{|V|}$.
\end{lemma}

\begin{proof}
We follow the proof of \cite[Theorem 1]{Woodall} due to Woodall.
We can assume without loss of generality that $G$ is connected (otherwise consider the components
separately).
The proof is by induction on $n$,  the number of  vertices of~$G$.
The base case, in which  $n=1$, is straightforward since $G$ has no loops.
Suppose  $n>1$.  
Since $G$ is planar, it has a vertex~$v$  
whose degree, $\ell$, is  between~$1$ and~$5$.
Let $e_1=(v,v_1),\ldots,e_\ell=(v,v_\ell)$ be the edges incident at~$v$. Let 
$\mathcal{A}_0 = \{A\subseteq E  \mid A\cap \{e_1,\ldots,e_\ell\}=\emptyset\}$ and
$\mathcal{A}_i = \{A \subseteq E \mid A\cap \{e_1,\ldots,e_i\}=\{e_i\}\}$.
Let $Z_i(G;q,w) = \sum_{A\in \mathcal{A}_i} w(A) q^{\kappa(V,A)}$, so $Z(G;q,w) = \sum_{i=0}^\ell Z_i(G;q,w)$.
It is easy to see (using the definition of the multivariate Tutte polynomial)
that $Z_0(G;q,w) = q Z(G-v;q,w)$ where, in the expression $Z(G-v,q,w)$, we view
$w$ as a weight function $w:E\setminus\{e_1,\ldots,e_\ell\}\rightarrow \mathbb{Q}$.
Also, for $i\in[\ell]$,
$Z_i(G;q,w) = w(e_i) Z(G'_i;q,w)$, where
$G'_i$ is the multigraph formed from~$G$ by deleting $e_1,\ldots,e_{i-1}$ 
and contracting~$e_i$
(i.e., identifying its endpoints and then deleting it). 
Note that   $G'_i$ may have parallel edges (though it has no loops).
However, if we consider two parallel edges $e$ and $f$
with weights $w(e)$ and $w(f)$, we know from Section~\ref{sec:weightImplement}
that the parallel composition of these two edges implements
the single edge weight $w^*=(1+w(e))(1+w(f))-1$.
Thus, we can replace these two parallel edges with a single edge~$e'$ with weight~$w^*$
without changing the value of the Tutte polynomial.
Also, note that since $|1+w(e)|\leq \rho$, $|1+w(f)|\leq \rho$ and $\rho\in(0,1)$,
we also have $|1+w^*|\leq \rho$.
We conclude that $Z(G'_i;q,w)=Z(G_i;q,w_i)$, where $G_i$ is the simple graph underlying~$G'_i$
and $w_i$ is the induced weight function, which is ``good'' in the sense that $|1+w_i(e)|\leq \rho$ for 
every edge $e$ of $G_i$.
The graph~$G_i$ has
vertex set $V-v$ and edge set $E_i=E\setminus\{e_1,\ldots,e_\ell\}\cup \{f_1,\ldots,f_{\ell_i}\}$,
where $\ell_i$ is the number of vertices in $v_{i+1},\ldots,v_\ell$ which
are not neighbours of $v_i$ in~$G$
and $f_1,\ldots,f_{\ell_i}$ are new edges connecting $v_i$ to these vertices.

 Let 
$\mathcal{B}_0 = \{A\subseteq E_i  \mid A\cap \{f_1,\ldots,f_{\ell_i}\}=\emptyset\}$ and
$\mathcal{B}_j = \{A \subseteq E_i \mid A\cap \{f_1,\ldots,f_j\}=\{f_j\}\}$.
Let $Z'_j(G_i;q,w_i) = \sum_{A\in \mathcal{B}_i} w_i(A) q^{\kappa(V-v,A)}$, 
so $Z(G_i;q,w_i) = \sum_{j=0}^{\ell_i} Z'_j(G_i;q,w_i)$.
Once again, $Z'_0(G_i;q,w) = Z(G-v;q,w_i)$.
Also, for $j\in[\ell_i]$, $Z'_j(G_i;q,w_i) = w_i(f_j) Z(G'_{i,j};q,w_i)$,
where $G'_{i,j}$ is the multigraph formed from~$G_i$ by deleting $f_1,\ldots,f_{j-1}$, 
and contracting~$f_j$.
As before, there is a simple graph~$G_{i,j}$ and a ``good'' weight function $w_{i,j}$
such that $Z(G'_{i,j};q,w_i) = Z(G_{i,j};q,w_{i,j})$. $G_{i,j}$~has vertex set $V-v$,
so we know by induction that $Z(G_{i,j};q,w_{i,j})\geq 0$. 

Putting all of the above together,
\begin{align*}
Z(G;q,w) &= q Z(G-v;q,w) + \sum_{i=1}^\ell  w(e_i) Z(G_i;q,w_i) \\
&= q Z(G-v;q,w) + \sum_{i=1}^\ell  w(e_i) \left(Z(G-v;q,w_i) +\sum_{j=1}^{\ell_i} w_i(f_j) Z(G_{i,j};q,w_{i,j})  \right).
\end{align*}

Since $w$ and $w_i$ are ``good'' weight functions, both $w(e_i)$ and $w_i(f_j)$ are at most~$0$.
Thus, $w(e_i) w_i(f_j) \geq 0$ so we get
\begin{align*}
Z(G;q,w) &\geq q Z(G-v;q,w) + \sum_{i=1}^\ell  w(e_i)  Z(G-v;q,w_i).\\
&= Z(G-v;q,w) \left( q + \sum_{i=1}^\ell  w(e_i)  \right).\\
&\geq Z(G-v;q,w)  \,\zeta,\\
\end{align*}
where the final inequality follows from $q>5(1+\rho)+\zeta$ and from the fact that the weight function~$w$ is good.
The result follows by induction.
\end{proof}

\section{Proving inapproximability}

\subsection{The starting point}

Our starting point is the following problem.
\decisionprob{\textsc{Planar cubic Maximum Independent Set}.}%
{A cubic planar graph $G$ and a positive integer $K$.}%
{Does $G$ contain an independent set of size at least $K$?}

\begin{lemma}
\textsc{Planar cubic Maximum Independent Set} is \textup{NP}-complete.
\end{lemma}

\begin{proof}
This problem is essentially the same as ``Node cover in planar graphs with 
maximum degree~3'', which was shown to be NP-complete by Garey and 
Johnson~\cite[Lemma~1]{GJ77}.  First, the complement
of a minimum node (or vertex) cover in a graph is a maximum independent 
set.  Thus Garey and Johnson's problem is the same as ``Maximum independent
set in a planar graph with maximum degree~3''.  So we just need to show that 
we can transform a planar graph with maximum degree~3 into a cubic graph
in such a way that the size of a maximum independent set changes in a controlled way.

It is easily checked that there is a (unique) simple 
planar graph~$T$ with degree sequence
$(1,3,3,3,3,3)$.  See Figure~\ref{fig:T}.
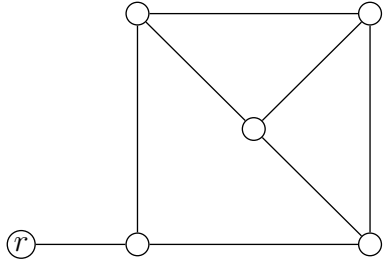
\begin{figure}
\centering{
\begin{tikzpicture}[fill=white,scale=0.45,
line width=0.5pt,inner sep=1pt,minimum size=2.5mm]
\pgfsetxvec{\pgfpoint{1.7cm}{0cm}}
\pgfsetyvec{\pgfpoint{0cm}{1.7cm}}
\path
(0,0) node [draw,fill,circle,minimum size=0.3cm](v1){$r$}
(2,0) node [draw,fill,circle,minimum size=0.3cm](v2){}
(6,0) node [draw,fill,circle,minimum size=0.3cm](v5){}
(4,2) node [draw,fill,circle,minimum size=0.3cm](v6){}
(2,4) node [draw,fill,circle,minimum size=0.3cm](v3){}
(6,4) node [draw,fill,circle,minimum size=0.3cm](v4){}
;
\draw[-] (v1)--(v2)--(v5)--(v4)--(v3)--(v2);
\draw[-] (v3)--(v6)--(v4);
\draw[-] (v5)--(v6);
\end{tikzpicture}

 }
 \caption{The graph $T$.}
 \label{fig:T}
 \end{figure}
Denote by~$r$ the unique vertex of degree~1.  Given a planar graph with 
maximum degree~3, we can form a planar cubic graph by attaching (via vertex~$r$)
the appropriate number of 
copies of~$T$ to the deficient vertices.  It is easily checked that each copy of~$T$
increases the size of a maximum independent set by~2.
\end{proof}

We will use the following  
variant of 
{\textsc{Planar cubic Maximum Independent Set}.}
This variant will help us to maintain planarity in our constructions.
\decisionprob{\textsc{Planar stretched cubic Maximum Independent Set}.}%
{A graph $G$ which is the 3-stretch of a cubic planar graph $H$
and a positive integer $K$.}%
{Does $G$ contain an independent set of size at least $K$?}

\begin{lemma}
 {\textsc{Planar stretched cubic Maximum Independent Set}} is \textup{NP}-complete.
\end{lemma}

\begin{proof}
 
Let $m'$ be the number of edges of $H$.
We claim that the size of a maximum independent set of $G$ is equal to $m'$ plus the
size of a maximum independent set of $H$ .
 
First, suppose that $H$ has an independent set of size $k$.
We use this independent set to 
construct an independent set of size $m'+k$ in $G$: For every IN-OUT edge of $H$
(that is, for every edge $(u,v)$ of $H$ such that $u$ is in the independent set, and $v$ is out),
the corresponding configuration in $G$ can be IN-OUT-IN-OUT\null.
For every OUT-OUT edge of $H$ the corresponding configuration of $G$ can be OUT-IN-OUT-OUT.

Next, suppose that $G$ has an independent set of size $m'+k'$ for some $k'\geq 0$.
We  construct an independent set of size $k'$ in $H$. 
Consider an independent set in $G$ of size $m'+k'$ which contains as many  degree-2 vertices as possible.
Consider the configuration corresponding to an edge of $H$.
It cannot be IN-OUT-OUT-IN, because one of the IN vertices could be moved to a degree-2 vertex
without changing the size of the independent set.
Thus, this independent set induces an independent set of $H$.
Since at most $m'$ degree-2 vertices are contained in the independent set,
the induced independent set in $H$ is size at least $k'$.
\end{proof}

\subsection{Some global variables}
\label{sec:2ii}

In our proofs, we will work with an instance~$G$ and~$K$ of 
 {\textsc{Planar stretched cubic Maximum Independent Set}} where $G$ has $n$~vertices and $m$~edges.
 For now, in order to do the preliminary work, let's view $n$, $m$ and $K$ as parameters corresponding to the
 size of the instance that we'll work with. Using the global constants from Section~\ref{sec:A3},
 we define the following quantities.

\begin{align*}
\nu&=3n-m-2K.\\
\epsilon &= 
\frac{{(B^-)}^3}{3 } 
\chi^{\nu}
2^{-(n+2m+4)}Q^{-3n+m}.\\
L &= |q^3| \eta \epsilon/2.\\
R&= {(B^-)}^3/3\epsilon.\\
\delta &=   
\frac{L^n \chi^\nu}{16\, A^* 5^n  M^n 2^{2m} Q^{9n}}
\end{align*} 
 
Since $G$ is the 3-stretch of a cubic planar graph, we will have $m=\frac98n$.
We will also assume that $K\leq\frac58n$, since this is an easy upper
bound on the size of any independent set in~$G$.
We will rely on the following inequalities, which follow from these considerations as long as $n$ is sufficiently large.
\begin{align}
 L&\leq 1\label{ineq:L1}\\
 R&\geq 1 \label{ineq:R1}\\
 \nu&\geq\frac58n\geq 1\label{ineq:nu1}\\
 0&<\delta<\epsilon < \chi \leq 1 \label{ineq:deltaepsilonchi}\\
 \delta &\leq \epsilon \eta / (6 A^*)\label{ineq:deltaepsAstar}
 \end{align}

 \subsection{The gadget revisited}
\label{sec:2iii}

Suppose that quantity $a$ satisfies (\ref{eqa1}), (\ref{eqa2}) and (\ref{eqa3}) with $\rho=\epsilon$
and that $b$ satisfies (\ref{eqb1}) and (\ref{eqb2}) with $\hat\rho=\delta$. Define
$c$, $d$ and $e$  via equations (\ref{eqc}), (\ref{eq:d}) and (\ref{eq:e}) respectively. 
Note that $d\in[-\delta,\delta]$ and $e\in[\epsilon,2 \epsilon]$.
Note also that~(\ref{eq:10Feb1}) implies
\begin{equation}
|Z_{0|12}| \leq \delta \mu.
\label{Z0|12upper}
\end{equation}

Now,
$1\leq A^*$, and
the constraints on $a$
imply 
$|3a+3a^2| \leq A^*$
and $|9 a^3 + 3 a^4 - 3 a q| \leq A^*$. Thus, using~(\ref{eqa3}),
the absolute value of each of the right-most three terms in (\ref{eq:10Feb2}) is
at most $\delta A^* /\eta$ and by (\ref{ineq:deltaepsAstar}), this is at most $\epsilon/6$.
Thus,
\begin{equation}
\label{eq:sep}
L
\leq 
|q^3 a^2(a+3)|\, \frac{\epsilon}{2}
  \leq
|Z_{0|1|2}|\leq 
|q^3 a^2(a+3)|\, 3\epsilon.
\end{equation}
Also, from (\ref{eq:10Feb3}),
\begin{equation}
 |qa^{2}(a+3)|\, {(B^-)}^{3} \leq |Z_{012}|\leq  \tau.
\label{ineq:Z012}
\end{equation}
Finally, we combine these to see
\begin{equation}
\label{ref:fugacity}
\frac{q^2 |Z_{012}|}{|Z_{0|1|2}|}
\geq 
\frac{q^2 |qa^{2}(a+3)|\, {(B^-)}^{3}}
{ |q^3 a^2(a+3)|\, 3\epsilon}
= R.
\end{equation}
We will also use the following quantity, defined in terms of the $Y$-gadget.
$$\Psi = { 
\left|
\frac{ q^2 Z_{012}}{Z_{0|1|2}}
\right|
}^{K-1} R\,
{|Z_{0|1|2}|}^{n}  {|q|}^{-3n} 
\chi^\nu.$$

 \subsection{The Main lemma}

 We can now state, and prove, our main lemma.

\begin{lemma}\label{lem:bedrock}
    Suppose $q \notin[0,5]$
 and that
$\alpha_1\notin [-2,0]$, $\alpha_2\in(-2,0)$ and $\alpha_3<-1$.
    Then there is no FPRAS for
    $\multitutte(q;\alpha_{1},\alpha_{2},\alpha_3)$ unless $\mathrm{RP}=\mathrm{NP}$.
\end{lemma}

\begin{proof}
Suppose $H$ is a cubic planar graph and $G$ is the 3-stretch of $H$.
Let $n=|V(G)|$ and $m=|E(G)|$
and note that Suppose $G$ and $K$ are inputs to
 {\textsc{Planar stretched cubic Maximum Independent Set}}.
Recall the definitions of the global variables from Section~\ref{sec:2ii}.
Our ultimate goal is to construct a planar instance $(G',w')$ of
$\multitutte(q;\alpha_{1},\allowbreak\alpha_{2},\alpha_3)$ such that 
a close approximation to $Z(G';q,w')$
enables us to determine whether $G$ has an independent set of size $K$.
To do this, we'll construct a weighted planar graph
$(\Ghat=(\Vhat,\Ehat),w)$
such that a close approximation to $Z(\Ghat;q,w)$  
enables us to determine whether $G$ has an independent set of size $K$,
 where
$w:\Ehat\to\{\beta,a,b\}$ 
and the edge-weight $a$ satisfies (\ref{eqa1}), (\ref{eqa2}) and (\ref{eqa3}) with $\rho=\epsilon$,
the edge-weight $b$ satisfies (\ref{eqb1}) and (\ref{eqb2}) with $\hat\rho=\delta$
and the edge-weight $\beta$ satisfies $|1+\beta|\leq \delta$.
Lemmas~\ref{lem:impa}, \ref{lem:impb} and \ref{lem:impbeta} show that such values can be implemented using
weights~$\alpha_1$, $\alpha_2$ and $\alpha_3$.
Thus, applying these implementations to the weight graph $\Ghat$ 
will give us $G'$ and $w'$ so that  $Z(\hatG;q,w)$ is an easily computable multiple of $Z(G',q,w')$, completing the proof.

Here is the construction of $\Ghat=(\Vhat,\Ehat)$. (See Figure~\ref{fig:graphpic}.)
First, fix any ordering on the vertices of $H$.
Next, let's set up some  useful notation for the graph $Y$, which we will use as a gadget.
A particular copy $Y^x$ of this gadget will have
vertex set  $V^x = \{\pr x0,\pr x1,\pr x2,\allowbreak
\pr x{\overline0},\pr x{\overline1},\pr x{\overline2}\}$
with vertices $\pr x0$, $\pr x1$ and $\pr x2$ arranged in clockwise order around
the outer face.
The edge set $E^x$ 
consists of three edges
$(\pr x0,\pr x{\overline0})$ , $(\pr x1,\pr x{\overline1})$, 
and $(\pr x2,\pr x{\overline2})$ of weight~$b$
and three edges 
$(\pr x{\overline0},\pr x{\overline1})$, $(\pr x{\overline1},
\pr x{\overline2})$ and $(\pr x{\overline2},\pr x{\overline0})$ of weight~$a$.
We will   construct a copy $Y^u$ of the $Y$-gadget for every vertex $u\in V(H)$.

Now, associate every edge $(u,v)$ of~$H$ with
two indices $i_{u,v}$ and $i_{v,u}$ in $\{0,1,2\}$ in such a way
that (a)~for every $u$, the three neighbors $v$ of $u$ get distinct indices $i_{u,v}$,
and (b)~the graph 
with vertex set $\bigcup_{u\in V(H)} V^u$ and
edge set 
$$
\bigcup_{u\in V(H)}E^u\,\,\cup\bigcup_{(u,v)\in E(H)} (\pr u{i_{u,v}},\pr v{i_{v,u}})
$$
is planar.
We will construct two copies, $Y^{uv}$ and  $Y^{vu}$, of the $Y$-gadget for every
edge $(u,v)$ of $H$. 
These correspond to the vertices of $G$ along the three-stretched edge $(u,v)$ of $H$.
Thus,
we have one $Y$ gadget for every vertex of $G$.

The vertex set $\Vhat$ is constructed 
from 
$\bigcup_{u\in V(H)} V^u \,\,\cup\,\, \bigcup_{(u,v)\in E(H)} (V^{uv} \cup V^{vu})$
by identifying some vertices,
In particular, for every
edge $(u,v)$ of $E(H)$ with $u<v$, 
identify  $\pr u{i_{u,v}}$ with $\pr{uv}0$. Also, identify $\pr{uv}2$ 
with $\pr{vu}1$.
Finally, identify $\pr{vu}0$ with $\pr v{i_{v,u}}$.
Note that 
$G$ has one vertex for each vertex of $H$ and two vertices for each 
edge of~$H$
so $n=|V(H)| + 2|E(H)|$. Also $m=3|E(H)|$.
So
$|\Vhat| = 6|V(H)| + 12|E(H)| - 3|E(H)| = 6n-m$.
 
Let 
$E = \bigcup_{u\in V(H)} E^u \,\,\cup\,\, 
\bigcup_{(u,v)\in E(H)} (E^{uv} \cup E^{vu})$. $E$ is all of the internal edges in the $Y$ gadgets.
So $|E| = 6|V(H)|+12|E(H)| = 6n$.
Let
$E'$ be the set of $m$ edges 
with weight~$\beta$ constructed as follows.
For each edge $(u,v)$ of $E(H)$ with $u<v$,
let $i= i_{u,v}-1\bmod 3$ and let $j=i_{v,u}-1\bmod 3$.
Add edges $(\pr ui,\pr{uv}1)$,
$(\pr{uv}1,\pr{vu}2)$ and $(\pr{vu}1,\pr vj)$ to $E'$. Let 
$\Ehat = E \cup E'$.
Note that $\hatG$ is planar.

\begin{figure}
\newcommand{\smoosh}[2]{\genfrac{}{}{0pt}{0}{#1}{#2}}
\centering{
\begin{tikzpicture}[fill=white,scale=0.7,
line width=0.5pt,inner sep=1pt,minimum size=2.5mm]
\pgfsetxvec{\pgfpoint{1.7cm}{0cm}}
\pgfsetyvec{\pgfpoint{0cm}{1.7cm}}
\path
(1.5,1) node [draw,fill,circle,minimum size=0.6cm](a){}
(1.5,-0.3) node [draw,fill,circle,minimum size=0.6cm](u1){\footnotesize$\pr u1$}
(1,2) node [draw,fill,circle,minimum size=0.6cm](b){}
(1,3.3) node [draw,fill,circle,minimum size=0.6cm](u2){\footnotesize$\pr u2$}
(2,2) node
[draw,fill,circle,minimum size=0.6cm](c){}
(3.5,2) node
[draw,fill,circle,minimum size=0.6cm](u0)
    {\footnotesize$\smoosh{\pr u0}{\pr {uv}0}$}
(5,2) node [draw,fill,circle,minimum size=0.6cm](d){}
(6,2) node [draw,fill,circle,minimum size=0.6cm](dd){}
(5.5,3) node [draw,fill,circle,minimum size=0.6cm](e){}
(5.5,4.3) node [draw,fill,circle,minimum size=0.6cm](uv1){\footnotesize$\pr{uv}1$}
(7.5,2) node [draw,fill,circle,minimum size=0.6cm](uv2)
   {\footnotesize$\smoosh{\pr{uv}2}{\pr{vu}1}$}
(9,2) node [draw,fill,circle,minimum size=0.6cm](f){}
(10,2) node [draw,fill,circle,minimum size=0.6cm](h){}
(9.5,3) node [draw,fill,circle,minimum size=0.6cm](g){}
(9.5,4.3) node [draw,fill,circle,minimum size=0.6cm](vu2){\footnotesize$\pr{vu}2$}
(11.5,2) node
[draw,fill,circle,minimum size=0.6cm](vu0)
   {\footnotesize$\smoosh{\pr{vu}0}{\pr v1}$}
(13,2) node [draw,fill,circle,minimum size=0.6cm](i){}
(13.5,4.3) node [draw,fill,circle,minimum size=0.6cm](v2){\footnotesize$\pr v2$}
(13.5,3) node [draw,fill,circle,minimum size=0.6cm](j){}
(14,2) node [draw,fill,circle,minimum size=0.6cm](k){}
(14,0.7) node [draw,fill,circle,minimum size=0.6cm](v0){\footnotesize$\pr v0$};
\draw[-] (b)--(c)--(u0)--(d)--(dd)--(uv2)--(f)--(h)--(vu0)--(i)--(k);
\draw[-] (b)--(a)--(c);
\draw[-] (a)--(u1);
\draw[-] (b)--(u2);
\draw[-] (d)--(e)--(dd);
\draw[-] (e)--(uv1);
\draw[style=dashed,color=red] (u2)--(uv1);
\draw[-] (f)--(g)--(h);
\draw[-] (g)--(vu2);
\draw[style=dashed,color=red] (uv1)--(vu2);
\draw[-] (i)--(j)--(k)--(v0);
\draw[-] (j)--(v2);
\draw[style=dashed,color=red] (v0)  .. controls (11,0.5) and (9,1) .. (uv2);
\end{tikzpicture}
}
\caption{The portion of $\hatG$ corresponding to edge $(u,v)$ of~$H$. 
In the picture, we assume $i_{u,v}=0$ and $i_{v,u}=1$. 
Edges of~$E$ are depicted as solid black lines and edges of~$E'$ 
are depicted as dashed red lines. Where two vertices have been identified, 
both the original labels have been displayed.}
\label{fig:graphpic}
\end{figure}
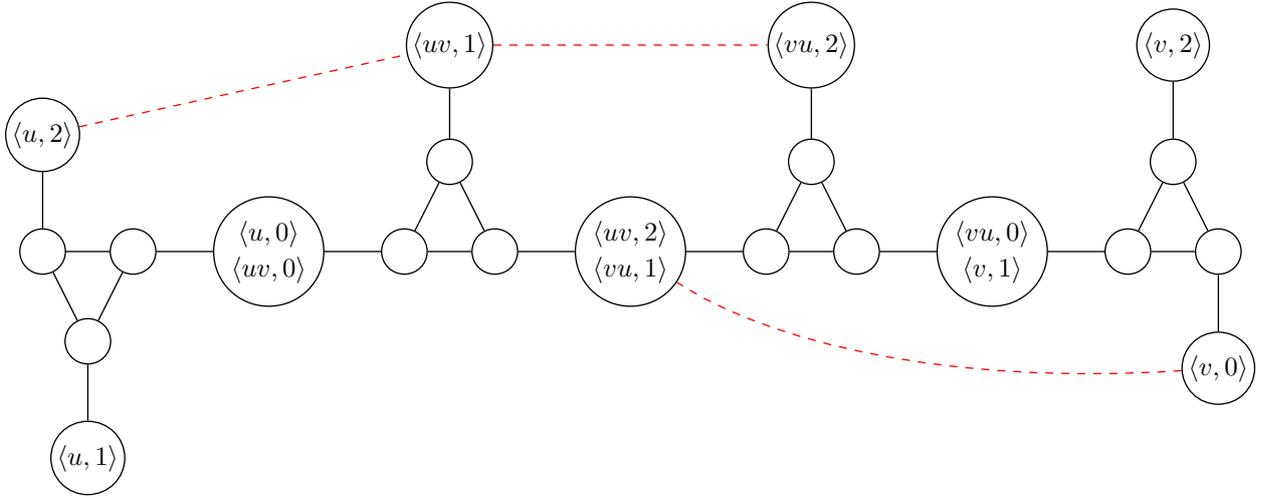

\def\part{\varPi}
The Tutte polynomial of $\hatG$ is given by
$$Z(\hatG;q,w) = \sum_{A \subseteq E} \sum_{B\subseteq E'} w(A)w(B) q^{\kappa(\hatV,A\cup B)},$$
where we have used the obvious fact $w(A\cup B)=w(A)w(B)$.  We would like to
go further and factor $\kappa(\hatV,A\cup B)$, in a similar way, but we
cannot do this directly because of the complex way that components 
in $(\hatV,A)$ and $(\hatV,B)$ may interact.  
To control this interaction, we partition sets $A\subseteq E$ according to the 
patterns of connectivities they induce within the various gadgets.
Specifically, let $\part=(S,D_0,D_1,D_2,T)$ be a labelled partition of 
$V(G)$ into five sets $S$, $D_0$, $D_1$, $D_2$ and $T$, some of which could be empty.
By ``labelled'' here, we mean that the five parts of the partition are distinguished by~$\part$.
In the following, it will help to think of $S$ as ``singleton'', 
$D$ as ``doubleton'' and $T$ as ``triple''.
Let $\mathcal{A}_\part$ denote the set of subsets $A\subseteq E$ such that the following
statements are true.
\begin{itemize}
\item For every $x\in S$, the vertices $\pr x0$, $\pr x1$ and $\pr x2$  
are in a single component of $(V^x,A\cap E^x)$.  Informally, all three vertices 
are connected {\it within\/} the gadget~$Y^x$.
\item For every $x\in D_i$, 
the vertex $\pr xi$ is in one component of $(V^x,A\cap E^x)$
and the other two vertices $\{\pr x0, \pr x1,\pr x2\}-\pr xi$ are in 
another.
\item For every $x\in T$, the vertices $\pr x0$, $\pr x1$ and $\pr x2$
are in three distinct components of $(V^x,A\cap E^x)$.
\end{itemize}

For a labelled partition $\part$ of $V(G)$ as above,
let $Z_\part$ be the contribution to $Z(\hatG;q,w)$ from edge sets
$A\in \mathcal{A}_\part$, specifically 
\begin{equation}\label{eq:Zpart}
Z_\part = \sum_{A \in \mathcal{A}_\part} 
   \sum_{B\subseteq E'} w(A)w(B) q^{\kappa(\hatV,A\cup B)}.
\end{equation}
It is clear that 
$Z(\hatG;q,w) = \sum_{\part} Z_\Pi$,
where $\part$ ranges over all labelled partitions $\part=(S,D_0,D_1,D_2,T)$
of $V(G)$ into five parts.
By constraining $A$ to come from a particular collection $\mathcal{A}_\part$
it now becomes possible to factor $\kappa(\hatV,A\cup B)$.  
To formalise this claim, let $V'= \hatV\setminus
\bigcup_{x\in V(G)} \{\pr x{\overline0},\pr x{\overline1},\pr x{\overline2}\}$, so
that $|V'| = |\hatV| - 3n = 3n - m$, and let $\varGamma$ denote the graph
$\varGamma = (V',E')$.   Suppose $\part=\{S,D_0,D_1,D_2,T)$ is some labelled 
partition, and denote by $\varGamma_\part$ the graph obtained from 
$\varGamma$ by identifying certain vertices.   Specifically, 
$\pr x0$, $\pr x1$ and $\pr x2$ are identified if $x\in S$,
$\pr x1$ and $\pr x2$ are identified if $x\in D_0$ (and symmetrically  
for $D_1$ and $D_2$), and none of the vertices are identified if $x\in T$.
 
With a view to factorising $\kappa(\hatV, A\cup B)$ into an $A$- and a $B$-part,
divide the connected components of $(\hatV, A\cup B)$ into two kinds:
those that contain no vertices in~$V'$ (and therefore are contained entirely within
a single $Y^x$),
and the others.  For convenience, let $D=D_0\cup D_1\cup D_2$.
The the number of connected components of the first kind is just
\begin{align}
&\sum_{x\in S}(\kappa(V^x,A\cap E^x)-1)+\sum_{x\in D}(\kappa(V^x,A\cap E^x)-2)
+\sum_{x\in T}(\kappa(V^x,A\cap E^x)-3)\notag\\
&\qquad = \sum_{x\in V}\kappa(V^x,A\cap E^x)-|S|-2|D|-3|T|.\label{eq:cc1st}
\end{align}
We argue that the connected components of the second kind are in 1-1 correspondence
with the connected components of $(V(\varGamma_\part),B)$.  Suppose two vertices $\pr xi$
and $\pr yj$ are connected by a path in $(\hatV,A\cup B)$;  then that same path can 
be traced out in $(V(\varGamma_\part),B)$ just by omitting the $A$-edges.
(Any pair of vertices joined by a 
sequence
of $A$-edges will have been 
identified in the construction of $\varGamma_\part$.)  Conversely, given a path in
$(V(\varGamma_\part),B)$, we can recover a path in $(\hatV,A\cup B)$ by interpolating
$A$-edges.  (We identify vertices in the construction of $\varGamma_\part$
only if they are in the same $A$-component.
Note that the ``recovered'' path may not be unique.
We conclude that the number of connected components of the second type
is $\kappa(V(\varGamma_\part),B)$.  Combining this with the count (\ref{eq:cc1st})
of connected components of the first type, we obtain 
$$
\kappa(\hatV,A\cup B)=\sum_{x\in V}\kappa(V^x,A\cap E^x)-|S|-2|D|-3|T|+\kappa(V(\varGamma_\part),B).
$$
Substituting for $\kappa(\hatV,A\cup B)$ in (\ref{eq:Zpart}) 
\begin{align*}
Z_\part &= \sum_{A \in \mathcal{A}_\part} 
   \sum_{B\subseteq E'} w(A)w(B) \bigg(\prod_{x\in V}q^{\kappa(V^x,A\cap E^x)}\bigg)\,
   q^{-|S|-2|D|-3|T|}\,q^{\kappa(V(\varGamma_\part),B)}\\
&= q^{-|S|-2|D|-3|T|}\bigg(\sum_{A \in \mathcal{A}_\part} 
   \prod_{x\in V}w(A\cap E^x)q^{\kappa(V^x,A\cap E^x)}\bigg)
   \bigg(\sum_{B\subseteq E'}w(B)q^{\kappa(V(\varGamma_\part),B)}\bigg).
\end{align*}
This immediately leads to the key identity
\begin{equation}\label{eq:key}
Z_{\part}
   =q^{-|S|-2|D|-3|T|}\,Z_{012}^{|S|}\,Z_{0|12}^{|D|}\,Z_{0|1|2}^{|T|}\,
   Z(\varGamma_{\part};q,\beta),
\end{equation}
where we recall that $\part=(S,D_0,D_1,D_2,T)$ and $D=D_0\cup D_1\cup D_2$.
Note that this identity captures the sought-for factorisation of $Z_{\part}$
into a part that is internal to the gadgets, 
and an part that is external, namely $Z(\varGamma_{\part};q,\beta)$.

With an eye on (\ref{eq:key}),
it is possible to give a short overview of the rest of the proof.
Recall that $Z(\hatG;q,w)$ is the sum over partitions $\part=(S,D_0,D_1,D_2,T)$
of~$Z_\part$.   If $D=D_0\cup D_1\cup D_2\not=\emptyset$, then $Z_\part$ is
negligible because $|Z_{0|12}|$ is tiny.  If $S$ is not an independent set in~$G$,
then $Z_\part$ is negligible because $\varGamma_{\part}$ has a loop,
and hence $Z(\varGamma_{\part};q,\beta)$ is tiny.  Finally, if $S$ is a 
maximum
independent set, then $Z_\part$ dominates because $|Z_{012}|$ is much larger than
$|Z_{0|1|2}|$.  So $Z(\hatG;q,w)$ is dominated by the contribution from maximum 
independent sets.

The rest of the proof is concerned with providing the estimates 
required to make the above proof sketch rigorous.
The number of vertices in $\varGamma_\part$ is
at most $3n-m$ and the number of edges is~$m$.  
Since $\max\{|\beta|,1\}\leq1+\delta\leq2$ and $|q|\leq Q$, we 
have the following general upper bound on 
$Z(\varGamma_{\part};q,\beta) = \sum_{B\subseteq E'} \beta^{|B|} q^{ \kappa(V(\varGamma_\part),B) }$: 
\begin{equation}\label{eq:ZGammaBd}
|Z(\varGamma_{\part};q,\beta)|\leq2^m(1+\delta)^mQ^{3n-m}\leq 2^{2m}Q^{3n-m}.
\end{equation}
If $\varGamma_\part$ has a loop we have the tighter bound
\begin{equation}\label{eq:withloop}
|Z(\varGamma_{\part};q,\beta)|\leq|1+\beta|\,2^{2m}Q^{3n-m}
   \leq \delta\,2^{2m}Q^{3n-m}.
\end{equation}
This comes about because the loop contributes $\beta$ when it is 
included and 1 when it is excluded, but the number of connected 
components is the same in both cases.
Recall the following general bounds on the other factors in~(\ref{eq:key})
which follow from (\ref{Z0|12upper}), (\ref{eq:sep}) and (\ref{ineq:Z012}):
\begin{align*}
|Z_{012}|,|Z_{0|12}|,|Z_{0|1|2}|&\leq\tau\\
|Z_{0|12}|&\leq \delta\mu<\tau\\
|q^{-|S|-2|D|-3|T|}|&\leq Q^{3n}.
\end{align*}

Following the proof sketch, first 
fix a partition~$\part$ in which $D$ is non-empty. 
Then, from (\ref{eq:key}), (\ref{eq:ZGammaBd}) and the bounds just noted,
$$
|Z_\part| \leq Q^{3n}\cdot \delta\mu\cdot\tau^{n-1}\cdot 2^{2m}Q^{3n-m}
\leq2^{2m}\delta\mu\tau^{n-1}Q^{6n}.
$$
So we get the following upper bound on contributions in which $D\not=\emptyset$:
\begin{equation}
\label{eq:D}
\sum_{\Pi:D\neq \emptyset}|Z_\Pi|\leq 5^n\cdot2^{2m}\delta\mu\tau^{n-1}Q^{6n} \leq \Psi/16.
\end{equation}
To see that the final inequality in (\ref{eq:D}) holds, first use (\ref{ref:fugacity})
to obtain $\Psi \geq R^K {|Z_{0|1|2}|}^{n}  {|q|}^{-3n} \chi^\nu$.
Then use  (\ref{ineq:R1}) 
($R^K\geq 1$) 
and  (\ref{eq:sep}) 
($|Z_{0|1|2}| \geq L$)
to see that this
is at 
least
$L^n {|q|}^{-3n} \chi^\nu$. Now 
divide the centre term in (\ref{eq:D}) by this lower bound for $\Psi$.
Plug in the definition of $\delta$  
and cancel the $\mu$ and $\tau$ in the numerator with $M$ in the denominator. 
The remaining terms cancel, and the result is at most~$1/16$. 

Next, fix a partition $\part$ in which $D$ is empty and $S$ is not an independent set of~$G$.
In this case, $\varGamma_\part$ has a loop, which arises from two adjacent gadgets
being contracted.  So from (\ref{eq:key}), (\ref{eq:withloop}) and the usual upper bounds,
$$
|Z_\part| \leq Q^{3n}\cdot\tau^n\cdot\delta\,2^{2m}Q^{3n-m}
\leq 2^{2m}\delta\tau^nQ^{6n}.
$$
So we get the following upper bound on contributions in which $S$ is not an independent set:
\begin{equation}
\label{eq:IS}
\sum_{\textstyle {\part:D=\emptyset,\atop S \text{ not independent}}}\!\!\!
|Z_\Pi|\leq 2^n\cdot2^{2m}\delta\tau^nQ^{6n} \leq \Psi/16.
\end{equation} 
The derivation of the final inequality in (\ref{eq:IS}) is essentially the same 
as the derivation of~(\ref{eq:D}). The
only difference is that a $5^n$ there has been replaced with an 
(even smaller) $2^n$ here. Also, a $\mu$ has
been replaced with a $\tau$ --- this still cancels against an $M$ as before.

Finally, fix a partition $\part$ in which $D$ is empty and $S$ is an independent set of~$G$
of size~$k$.  
From identity~(\ref{eq:key}), 
\begin{equation}
Z_{\part}=q^{-k-3(n-k)}\,Z_{012}^{k}\,Z_{0|1|2}^{n-k}\,
   Z(\varGamma_{\part};q,\beta)
=\Big(\frac{q^2Z_{012}}{Z_{0|1|2}}\Big)^kZ_{0|1|2}^n\,
   q^{-3n}\,Z(\varGamma_{\part};q,\beta).\label{eq:k}
\end{equation}
Thus, by (\ref{eq:ZGammaBd}),
$$
|Z_\part|\leq\Bigl|\frac{q^2Z_{012}}{Z_{0|1|2}}\Bigr|^k|Z_{0|1|2}|^n\,
   |q|^{-3n}\,2^{2m}Q^{3n-m}
$$
and by (\ref{ref:fugacity}) and (\ref{ineq:R1}),
\begin{equation}
\label{eq:small}
\sum_{\textstyle{\part:D=\emptyset, \, |S|<K,\atop 
S \text{ is independent}}}\!\!\!|Z_\part|
\leq 
2^n\cdot
{ 
\left|
\frac{ q^2 Z_{012}}{Z_{0|1|2}}
\right|
}^{K-1}
{|Z_{0|1|2}|}^{n}\, {|q|}^{-3n} \,
 2^{2m} 
Q^{3n-m}
\leq \Psi/16.
\end{equation}  
To see that the final inequality in (\ref{eq:small}) holds, plug in the definition of $\Psi$, and, inside that, plug in the
definition of~$R$, and inside that, plug in the definition of $\epsilon$. 
Everything cancels exactly.

Finally, we consider the situation $D=\emptyset$, $S$ is an independent set in~$G$,
and $|S|=K$.   Fix a partition $\part$ for which these conditions hold.
We are interested in obtaining a lower bound on~$|Z_\part|$.
Note that $|V(\varGamma_\Pi)|=\nu=3n-m-2K$.

\medskip\noindent{\it Case 1: $q>5$.\quad}
For  a lower bound 
on $|Z(\varGamma_\part;q,\beta)|$ and information about its sign, 
we use Woodall's Lemma~\ref{lem:Woodall}. Since  
$\chi = (q-5)/6$ and $\delta<\chi$, we have
$q > 5(1+\delta) + \chi$.
Furthermore, $|1+\beta|\leq \delta$.
Thus, Lemma~\ref{lem:Woodall} 
ensures 
that the sign of $Z(\varGamma_\part;q,\beta)$ is the
same for all $\part$ (it is always positive).
Also, we have shown
\begin{equation}
\label{eq:js}
|Z(\varGamma_\part;q,\beta)|\geq \chi^{\nu}.
\end{equation}
  
\medskip\noindent{\it Case 2: $q<0$.}  To determine the
same facts for $q<0$ we use~\cite[Theorem~4.1]{JacksonSokal}.
Note that $\varGamma_\part$ has no loops. 
Let $C_1,\ldots,C_\nu$ denote
the coefficients of $Z(\varGamma_\part;q,\beta)$, viewed as a polynomial in~$q$,
so $Z(\varGamma_\part;q,\beta) = \sum_{j=1}^{\nu} C_j q^j$.
Let $\pi_j=1$ if $C_j> 0$, $\pi_j=0$ if $C_j=0$ 
and $\pi_j=-1$ if $C_j<0$. 
Then
$$Z(\varGamma_\part;q,\beta) = {(-1)}^\nu
 \sum_{j=1}^{\nu}  {(-1)}^{\nu-j} \pi_j \,|C_j|\, {|q|}^j.$$
Jackson and Sokal~\cite[Theorem 4.1]{JacksonSokal} showed (assuming $\delta\leq 1$
which holds by inequality (\ref{ineq:deltaepsilonchi}))  that 
$(-1)^{\nu-j} \pi_j \geq 0$.
So 
$$Z(\varGamma_\part;q,\beta) = {(-1)}^\nu \sum_{j\in\{1,\ldots,\nu\}, \pi_j\neq 0}  |C_j| \,{|q|}^j.$$
Note that for $j=\nu$, $C_j=1$
so Equation~(\ref{eq:js}) holds and
the sign of 
$Z(\varGamma_\part;q,\beta)$ is the
same for all partitions~$\part$ in which $D=\emptyset$ and $S$ is an independent set of size~$K$
(the sign depends on the parity of~$\nu$).
This concludes Case~2.\footnote{ Establishing (\ref{eq:js}) is the main barrier to extending our result to $q\in[0,5]$.
The Tutte polynomial with $\beta$ close to $-1$ is similar to the chromatic polynomial.
Essentially, we are using the fact that this polynomial is non-zero with sign $(-1)^\nu$ when $q<0$ and is
positive when $q>5$. There are known to be many zeroes of   chromatic polynomials in between~$0$ and~$5$.}

Now, for a partition $\part$ in which $D$ is empty and $S$ is an independent set of~$G$
of size~$K$, Equations~(\ref{eq:k}),
(\ref{ref:fugacity}) and (\ref{eq:js})
give 
\begin{equation}
\label{eq:K}
|Z_\part|  \geq
  \Psi.
\end{equation}
Since the sign of $Z(\varGamma_\part;q,\beta)$ is the same for all 
$\part$ under consideration,  
it is apparent from~(\ref{eq:k}) that the sign of $Z_\part$ 
depends only on the sign of~$q$, the sign of $Z_{012}$, the
sign of $Z_{0|1|2}$ and the parity of $K$ and~$n$. It does not depend 
on the set~$S$.   

So if $G$ has $N>0$ independent sets of size~$K$
then by Equations~(\ref{eq:D}), (\ref{eq:IS}), 
(\ref{eq:small}) and~(\ref{eq:K}),
$|Z(\Ghat;q,w)| \geq N \Psi - 3\Psi/16 \geq 3\Psi/4$.
On the other hand, if $G$ has no independent sets of size~$K$
then the same equations give
$|Z(\Ghat;q,w)| \leq 3 \Psi/16 < \Psi/4$.
So if we could approximate $Z(\Ghat;q,w)$ within a factor of $\tfrac32$
then we could determine whether or not $G$ has an independent set of size~$K$.
\end{proof}

\subsection{The  main result}

\begin{theorem}\label{thm:shift}
Suppose $(x,y)\in \mathbb{Q}^2$ satisfies
$q=(x-1)(y-1) \notin[0,5]$  .  Suppose also that
it is possible to shift the point
$(x,y)$ to a point $(x_1,y_1)$ with $y_1\notin[-1,1]$
and to a point $(x_2,y_2)$ with $y_2\in(-1,1)$ and to 
a point $(x_3,y_3)$ with $y_3<0$.
Then there is no FPRAS for
$\tutte(x,y)$
unless
$\mathrm{RP}=\mathrm{NP}$.
\end{theorem}

\begin{proof}
This follows easily from Lemma~\ref{lem:bedrock} and is similar to the proof of
\cite[Theorem 2]{tuttepaper}. 
For completeness, here is a proof.

Let $\alpha=y-1$ and $\alpha_i = y_i - 1$.
Let $\Upsilon_i$ be a planar graph 
with distinguished vertices $s_i$ and $t_i$
that shifts $(x,y)$ to $(x_i,y_i)$.
Note that $\Upsilon_i$
shifts $(q,\alpha)$ to $(q,\alpha_i)$.

Suppose $(G,w)$ is an instance of $\multitutte(q;\alpha_{1},\alpha_{2},\alpha_3)$
and note that $\alpha_1$, $\alpha_2$ and~$\alpha_3$ satisfy the conditions
of Lemma~\ref{lem:bedrock}. Suppose that $G$ has $m_i$ edges with weight $\alpha_i$.
Denote by $\hatG$ the graph derived from $G$ by applying the above shifts --- replacing each
edge with weight $\alpha_i$ with a copy of $\Upsilon_i$ (using the distinguished vertices $s_i$ and $t_i$).
Let $\hatw$ be the constant weight function which assigns weight $\alpha$ to
every edge in $\hatG$.
Then by Equation~(\ref{eq:shift}),
 $$Z(\hatG;q,\hatw)
 = 
{\left(\frac{Z_{s|t}(\Upsilon_1)}{q^2}\right)}^{m_1}
{\left(\frac{Z_{s|t}(\Upsilon_2)}{q^2}\right)}^{m_2}
{\left(\frac{Z_{s|t}(\Upsilon_3)}{q^2}\right)}^{m_3}
Z(G;q,w),$$
so by Equation~(\ref{eq:rcequiv}),
$$
(y-1)^n (x-1)^{\kappa} T(\hatG;x,y) =
{\left(\frac{Z_{s|t}(\Upsilon_1)}{q^2}\right)}^{m_1}
{\left(\frac{Z_{s|t}(\Upsilon_2)}{q^2}\right)}^{m_2}
{\left(\frac{Z_{s|t}(\Upsilon_3)}{q^2}\right)}^{m_3}
Z(G;q,w),$$ 
where 
$n$ is the number of vertices in $\hatG$, and $\kappa$ is the number of
connected components in $\hatG$.
Note that $Z_{s|t}(\Upsilon_i)\neq 0$ since
$q Z_{st}(\Upsilon_i)/Z_{s|t}(\Upsilon_i)=\alpha_i$.
Thus an FPRAS for $\tutte(x,y)$
would yield an FPRAS
for the problem
$\multitutte(q;\alpha_{1},\alpha_{2},\alpha_3)$, contrary to Lemma~\ref{lem:bedrock}.
\end{proof}

 The following corollary  
identifies regions where approximating $\tutte(x,y)$ is intractable.
It is illustrated in Figure~\ref{fig:two}.

\begin{corollary}\label{cor:x<-1}
Suppose $\mathrm{RP}\not=\mathrm{NP}$. Then
there is no FPRAS for $\planartutte(x,y)$
when $(x,y)$ is a point 
in the following regions, where $q$ denotes
$(x-1)(y-1)$:
\begin{enumerate}
\item $x<0$, $y<0$ and $q>5$;
\item $x>1$, $y<-1$;
\item $y>1$, $x<-1$.
\end{enumerate}
\end{corollary}

\begin{proof} 

We will show that for each point $(x,y)$ in the following regions,
we can shift to points $(x_1,y_1)$, $(x_2,y_2)$ and $(x_3,y_3)$
satisfying the conditions of Theorem~\ref{thm:shift}.

For the remaining cases, we use
the fact that, when $G$ is a planar graph and $G^*$ is any plane dual of $G$,
$T(G;x,y) = T(G^*;y,x)$ \cite[\S3.3.7] {Welsh93} (so the fact that there is no FPRAS at $(x,y)$
implies that there is no FPRAS at $(y,x)$ and vice-versa).
The regions that we consider are as follows. 

\begin{enumerate}
\item {\it $x<-1$, $y<-1$ and $q>5$}:\quad
We can take $(x_3,y_3)$ and $(x_1,y_1)$ to be $(x,y)$ since $y<-1$.
We can
realise $(x_2,y_2)$ using a large, odd, $k$-stretch
so $y_2 = q/(x^k-1)+1$ which is in the range $(-1,1)$.

\item {\it $-1\leq x < 0$, $y\leq -3/2$ and $q>5$}:\quad
Note that the condition $y\leq -3/2$ is implied by the
bounds on $x$ and $q$.
As above, we can take $(x_3,y_3)$ and $(x_1,y_1)$ to be $(x,y)$ since $y<-1$.
Next, realise $(x',y')$  using a 2-thickening   so
$$x' = \frac{q}{y^2-1}+1>1.$$
Choose $j$ so that ${x'}^j>q/|x|$. Realise $(x_2,y_2)$ by
$j$-stretching $x'$ and combining this in series with $x$ so $x_2 = {x'}^j x < -q$.
Since $|x_2-1|\geq q$, we have $|y_2-1|\leq 1$. Also $y_2-1 <0$ Thus $y_2 \in (0,1)$.

\item {\it $x>1$, $y<-1$}:\quad
Note that $q<0$. Again, we can take 
both $(x_1,y_1)$ and $(x_3,y_3)$ to be the point
$(x,y)$ since that gives $y_1=y_3<-1$. 
We get to $(x_2,y_2)$ by a $j$-stretch, for sufficiently large $j$.
This gives 
$$y_2 = \frac{q}{x^j-1}+1\in(-1,1).$$

\end{enumerate}
\end{proof}

\subsection{The deferred proofs from Section  \ref{sec:A4}}
\label{sub:deferred}

In this section we provide the proofs of Lemmas~\ref{lem:impa}, \ref{lem:impb} and \ref{lem:impbeta}.
We start with some technical lemmas which show that
we can use $\alpha_1$, $\alpha_2$ and $\alpha_3$ to implement a very close approximation to any target weight~$T^*$, provided 
$T^*\notin[-2,0]$.

\begin{lemma}  \label{lem:shiftbigqbigT}
Suppose $q >5$
 and that
$\alpha_1\notin [-2,0]$, $\alpha_2\in(-2,0)$ and $\alpha_3<-1$.
Suppose that $T^-$ and $T^+$ satisfy 
$0<T^-\leq T^+$.
Given a target edge-weight $T^*\in[T^-,T^+]$ and a
positive value~$\pi$ which is sufficiently small with respect to~$q$, $\alpha_1$, $\alpha_2$,
$\alpha_3$, $T^-$, and $T^+$,
there is a planar graph $\Upsilon$ (depending on $T^*$ and $\pi$) and a weight 
function $\hatw:E(\Upsilon)\rightarrow \{\alpha_1,\alpha_2,\alpha_3\}$ 
that implements a weight $w^*$
with 
$T^* - \pi \leq w^* \leq T^*$.
The size of $\Upsilon$ is at most a polynomial in $\log(\pi^{-1})$.
(This upper bound on the size of $\Upsilon$ does not depend on~$T^*$, though it
does depend on the fixed bounds~$T^-$ and~$T^+$.)
\end{lemma}

\begin{proof}

The weights that we have available for our implementations are $\alpha_1$, $\alpha_2$  and $\alpha_3$
and the target edge weight is~$T^*$.
It will be useful to use $(x,y)$ coordinates as well as $(q,\alpha)$ coordinates 
since series compositions power~$x$ and parallel compositions power~$y$.
Recall that the relationship between the two coordinate systems is given by $q=(x-1)(y-1)$ and
$\alpha=y-1$.
Thus, 
we define   
\begin{itemize}
\item $y'_i = 1+\alpha_i$ for $i\in\{1,2,3\}$,
\item $T = 1+ T^*$, and
\item  $x'_i=q/(y'_i-1)+1$ for $i\in\{1,2,3\}$.
\end{itemize}
(The primes are just there because we use the notation $y_i$ for something else below.)
 
We will show how to 
use the values $y'_1$ , $y'_2$ and $y'_3$ to implement an edge-weight whose $y$-coordinate is between
$T-\pi$ and $\pi$.
(In fact, we won't use $y'_2$ in the proof of this lemma, but we will use it in the proof of some of the related lemmas.)
To do this efficiently (keeping the size of $\Upsilon$ at most a polynomial in 
$\log(\pi^{-1})$) we
need to be  somewhat careful about decomposing~$T$.  
Let $(x_1,y_1)$ be the point on the hyperbola $(x_1-1)(y_1-1)=q$ given by $y_1 = {y'_1}^2$.
Note that $y_1>1$ so $x_1>1$ and that we we can implement $(x_1,y_1)$ by $2$-thickening from $(x'_1,y'_1)$.

Let
$$y_j = \frac{q}
{
{ x_1
}^j
-1
}
+1.$$
Let $x_j$ be the corresponding  value so that $(x_j-1)(y_j-1)=q$.
Note that, for every integer $j\geq 1$,
we can implement $(x_j,y_j)$ by   $j$-stretching from $(x_1,y_1)$.
Also, since $x_1>1$, we have $y_j>1$ and $y_j>y_{j+1}$.

Now, for every integer $j\geq 1$, we recursively define a quantity
$d_j$ in terms of the values of $d_1,\ldots,d_{j-1}$.
In particular, we first find the largest power of~$y_1$ not exceeding~$T$,
and divide~$T$ by this power to obtain~$d_1$;
then we divide~$d_1$ by the largest power of~$y_2$ to obtain~$d_2$, and so on.
Formally,
$$d_j = \left\lfloor
\frac{\log(T 
\prod_{\ell=1}^{j-1}{y_\ell}^{-d_\ell}
)
}
{\log(y_j)}
\right\rfloor.$$

Let $y''_m = \prod_{\ell=1}^{m} y_\ell^{d_\ell}$. Note that $T/y_m \leq y''_m \leq T$.
Also, since $d_j$ is a non-negative integer,
we can implement $y''_m$ 
with a graph $\Upsilon_m$
by 
$d_\ell$-thickening $y_\ell$ (for $\ell\in\{1,\ldots,m\}$) and then combining these in parallel.

Let 
$$m = \left\lceil 
\frac
{\log( q T /\pi+1)}{\log(
 x_1)}\right\rceil.$$
Note that $y_m \leq 1+\pi/T \leq 1/(1-\pi/T)$,
so $1/y_m \geq 1-\pi/T$.
Let $y=y''_m$ and let $\Upsilon=\Upsilon_m$.
Note that $T-\pi \leq y \leq T$, as required.
 
To see that this implementation is feasible,  note that $m$ 
is not too large. In particular, 
for fixed $q$, $y'_1$, $T^-$ and $T^+$,  
$m$
is bounded from above by a polynomial
in the logarithm of $\pi^{-1}$.
To finish, we must show that the same is true of $d_1,\ldots,d_m$.
Here, the key observation is
that $y_j^{d_j} \leq T/{y''_{j-1}} \leq y_{j-1}$,
so
$d_j \leq \log(y_{j-1})/\log(y_j)$.
Then for $y_j\leq 5/4$, say,
we have $\tfrac34 (y_j-1) \leq \log(y_j) \leq y_j-1$ which suffices.
\end{proof}

 \begin{lemma}  \label{lem:shiftsmallqbigT}
Suppose $q <0$
 and that
$\alpha_1\notin [-2,0]$, $\alpha_2\in(-2,0)$ and $\alpha_3<-1$.
Suppose that $T^-$ and $T^+$ satisfy 
$0<T^-\leq T^+$.
Given a target edge-weight $T^*\in[T^-,T^+]$ and a
positive value~$\pi$ which is sufficiently small with respect to~$q$, $\alpha_1$, $\alpha_2$,
$\alpha_3$, $T^-$, and $T^+$,
there is a planar graph $\Upsilon$ (depending on $T^*$ and $\pi$) and a weight 
function $\hatw:E(\Upsilon)\rightarrow \{\alpha_1,\alpha_2,\alpha_3\}$ 
that implements a weight $w^*$
with 
$T^* - \pi \leq w^* \leq T^*$.
The size of $\Upsilon$ is at most a polynomial in $\log(\pi^{-1})$.
\end{lemma}

\begin{proof}
The situation is the same as that of Lemma~\ref{lem:shiftbigqbigT} except that $q<0$.
The proof
is very similar to the proof of Lemma~\ref{lem:shiftbigqbigT}, and we use the notation from that proof.
Here, we start by implementing a point $(x_1,y_1)$ with
$x_1<-1$. 
Then we just use odd values of $j$ 
and it suffices to take
$$m = \left\lceil 
\frac
{\log( |q| T /\pi)}{\log(
 |x_1|)}\right\rceil,$$
 and to follow the proof of Lemma~\ref{lem:shiftbigqbigT}.

The point $(x_1,y_1)$ is reached as follows.  
If ${y'_1}^2 < 1+|q|/2$ then
we can take $(x_1,y_1)= 
\big(q/({y'_1}^2-1)+1,{y'_1}^2\big)$
since $x_1<-1$. 
Otherwise, proceed as follows.
Let $$\xi=\frac{|q|}{2}\,\frac{1}{1+|q|/2}.$$
Choose a positive integer $j$ so that
$$-\xi < \frac{q}{\big(
q/(y'_2-1)+1
\big)^j-1} < 0.$$
There is such a $j$ since $y'_2\in(-1,1)$.
Now let  $(\hatx,\haty)$ be
the $j$-stretch of $(x'_2,y'_2)$ so  $1-\xi < \haty < 1$.
Now let
$$k = 1 + \left\lfloor
\frac{\log\left(
(1+|q|/2)/{y'_1}^2
\right)}
{\log(\haty)} 
\right\rfloor.$$
Note that $k$ is a positive integer since ${y'_1}^2 \geq 1+|q|/2$.
Let 
$(x_1,y_1)$ be the parallel composition of 
$\big(q/({y'_1}^2-1)+1,
{y'_1}^2\big)$ with the $k$-thickening of $(\hatx,\haty)$.
Thus, $y_1 = {\haty}^k {y'_1}^2$.
Note that $1< \haty (1+|q|/2) \leq
y_1 < 1+|q|/2$  so
$x_1 < -1$.
\end{proof}

\begin{lemma}  \label{lem:shiftsmallT}
Suppose $q \notin[0,5]$
 and that
$\alpha_1\notin [-2,0]$, $\alpha_2\in(-2,0)$ and $\alpha_3<-1$.
Suppose that $T^-$ and $T^+$ satisfy 
$2< T^-\leq T^+$.
Given a target edge-weight $T^*$ with $-T^*\in[T^-,T^+]$ and a
positive value~$\pi$ which is sufficiently small with respect to~$q$, $\alpha_1$, $\alpha_2$,
$\alpha_3$, $T^-$, and $T^+$,
there is a planar graph $\Upsilon$ (depending on $T^*$ and $\pi$) and a weight 
function $\hatw:E(\Upsilon)\rightarrow \{\alpha_1,\alpha_2,\alpha_3\}$ 
that implements a weight $w^*$
with 
$T^* \leq w^* \leq T^*+\pi$.
The size of $\Upsilon$ is at most a polynomial in $\log(\pi^{-1})$.
\end{lemma}

\begin{proof}
Once again, we use the notation from the proof of Lemma~\ref{lem:shiftbigqbigT}. The
situation is the same as 
that of Lemmas~\ref{lem:shiftbigqbigT} and~\ref{lem:shiftsmallqbigT} except that  
the target edge weight is negative (in fact, it is less than~$-2$). 
 
Choose an even positive integer~$j$ so that $y_3' {y'_2}^j \in (-1,0)$
and let $\hat y = -y_3' {y'_2}^j$ (so $\hat y \in (0,1)$).
Recall that $T = T^*+1$.
Now let 
$U^* = \frac{-T}{\hat y}-1$.
Note that $$U^* \in \left[\frac{T^--1}{\hat y} -1,\frac{T^+-1}{\hat y}-1\right]$$ and
that the lower bound $(T^--1)/\hat y -1$
is positive.
Using Lemma~\ref{lem:shiftbigqbigT} or~\ref{lem:shiftsmallqbigT} (whichever is appropriate, depending on the sign of~$q$)
with target edge weight $U^*$ and error value~$\pi/\hat y$,
implement an edge-weight whose $y$-coordinate  $y'$ 
satisfies
$$\frac{-T}{\hat y} - \frac{\pi}{\hat y} \leq y' \leq \frac{-T}{\hat y}.$$
Then take $y'$ in parallel with $y'_3$ 
and $j$ copies of $y'_2$
to get a value $y = - y' \hat y$ satisfying
$T \leq y \leq T+\pi$.
 \end{proof}

We can now provide the proof of Lemmas~\ref{lem:impa}, \ref{lem:impb} and \ref{lem:impbeta}.
For convenience, we re-state these lemmas here.

\newcounter{eqtemp}
\newtheorem*{lemma1}{Lemma 1}
\begin{lemma1} 
\setcounter{eqtemp}{\value{equation}}
\setcounter{equation}{4}
Suppose $q\notin [0,5]$
 and that
$\alpha_1\notin [-2,0]$, $\alpha_2\in(-2,0)$ and $\alpha_3<-1$.
Given a positive 
constant~$\rho$
which is  sufficiently small with respect to~$q$, $\alpha_1$, $\alpha_2$, and
$\alpha_3$, 
there is a planar graph $\Upsilon$ (depending on  
$\rho$) and a weight 
function $\hatw:E(\Upsilon)\rightarrow \{\alpha_1,\alpha_2,\alpha_3\}$ 
that implements a weight 
$a$, such that
\begin{gather}
  A^- \leq |a| \leq A^+,  \\
  q + \rho  < f(a) \leq q + 2 \rho, \\
  | f(a)|\geq  \eta. 
\end{gather}
The size of $\Upsilon$ is at most a polynomial in 
$\log(\rho^{-1})$.
\setcounter{equation}{\value{eqtemp}}
\end{lemma1}

\begin{proof}
First, suppose $q>5$. 
We start by noting that any $a$ that satisfies (\ref{eqa2})
also has $\frac12 \leq a \leq q$, so $A^- \leq |a| \leq A^+$
and $| f(a)|>\eta$. So we just need to see how to implement a value of $a$
that satisfies~(\ref{eqa2}).
We will use Lemma~\ref{lem:shiftbigqbigT} with the target value $T^*$
being the solution to the equation $f(T^*)=q+2\rho$
and the error value $\pi = \rho^2$.
As noted above, $T^*\in[A^-,A^+]$.
To see that $f(T^*-\rho^2)> q+\rho$,
note that
$$f(T^*) - f(T^* - \rho^2) = 
(6T^* + 3 (T^*)^2) \rho^2 - (3 + 3 T^*) \rho^4 + \rho^6.$$
This is at most $\rho$, as required, as long as $\rho$ is sufficiently small with respect to~$T^*$ (which is in between
$A^-$ and $A^+$, which depend only on~$q$).
Since $\rho$ is assumed to be sufficiently small with respect to~$q$, this is the desired result.
  
Now suppose $q<0$.
Start by noting that if $y$ satisfies
$f(-3-y) \leq q+2\rho \leq q/2$ then
$y\geq \eta$. Also, if $y$ satisfies
$f(-3-y) \geq q+\rho\geq q$ then
$y\leq y^*$. Thus, if $a=-3-y$ satisfies (\ref{eqa2}) then  $-a \in [A^-,A^+]$.
We will now argue that these conditions also imply $|f(a)|>\eta$.
To see this, check that $f(-3-x/2)+x/2$ is negative for $x>0$.
Therefore, taking $x=-q$, we get $f(-3+q/2)<q/2$. By the definition of $\eta$, we find that $q/2<-\eta$.
Thus, for $-a\in[A^-,A^+]$, we have $f(a)\leq q/2 < -\eta$, as required.
So we just need to see how to find a value of $a$
that satisfies~(\ref{eqa2}). This is now essentially the same as  the $q>5$ case except that
we use Lemma~\ref{lem:shiftsmallT}.
 \end{proof}

\newtheorem*{lemma2}{Lemma 2}
\begin{lemma2} 
\setcounter{eqtemp}{\value{equation}}
\setcounter{equation}{7}
Suppose $q\notin [0,5]$
 and that
$\alpha_1\notin [-2,0]$, $\alpha_2\in(-2,0)$ and $\alpha_3<-1$.
Suppose, for a positive value~$\rho$,
which is  sufficiently small with respect to~$q$, $\alpha_1$, $\alpha_2$, 
and $\alpha_3$,
the value~$a$ satisfies inequalities (\ref{eqa1}), (\ref{eqa2}) and~(\ref{eqa3}).
Let 
\begin{equation}
  c = a^2+3a+q
\end{equation}
Given a positive 
constant~$\hat\rho$
which is sufficiently small with respect to~$q$, $\alpha_1$, $\alpha_2$, and $\alpha_3$,
there is a planar graph $\Upsilon$ (depending on  
$\hat \rho$) and a weight 
function $\hatw:E(\Upsilon)\rightarrow \{\alpha_1,\alpha_2,\alpha_3\}$ 
that implements a weight 
$b$, such that
\begin{gather}
B^- \leq |b| \leq B^+, \>\mbox{and } \\
  - \hat \rho \leq b + c \leq \hat \rho. 
\end{gather}
The size of $\Upsilon$ is at most a polynomial in 
$\log({\hat \rho}^{-1})$.
\setcounter{equation}{\value{eqtemp}}
\end{lemma2}
 
 \begin{proof}

First, suppose $q>5$.
Note that   $B^- \leq c - 1\leq c - \hat \rho$ and
$c+\hat \rho \leq c+1\leq B^+$
so it suffices to implement a value $b$ with
$-c-\hat \rho \leq b \leq -c + \hat \rho$.
 So we use 
Lemma~\ref{lem:shiftsmallT}   choosing target~$T^*=-c$ and
$\pi=\hat \rho$.  From our observation above, $-T^*\in[B^-+1,B^+-1]$.
  
Next, suppose $q<0$.
By equation~(\ref{eqa2}) (using the fact that $a<0$),
$$\frac{q+2\rho}{a} \leq a^2 + 3 a \leq \frac{q+\rho}{a},$$
so since $a\leq -3$,
$$-\left(
\frac{4}{3}|q|+1\right) \leq
\frac{q+2}{a}+q <
\frac{q+2\rho}{a} + q \leq
a^2 + 3 a + q \leq \frac{q+\rho}{a}+q \leq q
\left(1+\frac{1}{a}\right)
\leq - \frac{2}{3}|q|
< 0.$$
Thus, $-c-\hat\rho \geq B^-$ and $-c+\hat\rho \leq B^+$.
So it suffices to implement a value $b$ with
$-c-\hat\rho \leq b \leq -c + \hat\rho$.  
For this, we just use the argument in Lemma~\ref{lem:shiftsmallqbigT} 
 choosing target $T^*=-c$ and
$\pi=\hat\rho$.
\end{proof}
 
 {\bf Lemma \ref{lem:impbeta}}
 {\sl
Suppose $q\notin [0,5]$
 and that $\alpha_2\in(-2,0)$.
Given a positive 
constant~$\rho$
which is  sufficiently small with respect to~$q$ and $\alpha_2$, there is a planar graph $\Upsilon$ (depending on  
$\rho$) and a weight 
function $\hatw:E(\Upsilon)\rightarrow \{ \alpha_2 \}$ 
that implements a weight 
$\beta$, such that $|1+\beta|\leq \rho$.
The size of $\Upsilon$ is at most a polynomial in 
$\log(\rho^{-1})$.} 

\begin{proof}
This is already done in~\cite{tuttepaper}.
To implement $\beta$, choose a positive integer $k$ such that
$|(\alpha_2+1)^k|<\rho$ then implement $\beta$ by $k$-thickening
$\alpha_2$.  \end{proof}

\section{The lower branch of $q=3$}

The following is NP-hard \cite{GJS}.
\decisionprob{\textsc{Planar $3$-Colouring}.}%
{A  planar graph $G$.}%
{Does $G$ have a proper $3$-colouring?}

The following lemma gives hardness for approximating the Tutte polynomial on the
lower branch of the $q=3$ hyperbola. See Figure~\ref{fig:two}.
\begin{lemma}
\label{lem:q=3}
Suppose $\mathrm{RP}\not=\mathrm{NP}$. Then
there is no FPRAS for $\tutte(x,y)$
when $(x,y)$ satisfies $(x-1)(y-1)=3$ and $x,y < 1$.
\end{lemma}

\begin{proof}

We will consider a point $(x,y)$ with $-1<y<1$.
The remaining cases follow by symmetry between~$x$ and~$y$ as in the proof
of Corollary~\ref{cor:x<-1}.
Let $G=(V,E)$ be an input to \textsc{Planar $3$-Colouring} with $n$ vertices.
For an even positive integer~$k$, let $G^k$ be the graph formed from $G$ by $k$-thickening every edge and let $E^k$
be its edge set.
It is well-known (see, for example, \cite[Section 5.1]{tuttepaper}) that, assuming $(x-1)(y-1)=3$,
$$ T(G^k;x,y) = {(y-1)}^{-n}{(x-1)}^{-\kappa(V,E^k)} 
\sum_{\sigma:V\rightarrow\{1,2,3\}}
y^{\mathrm{mono}(\sigma)},
$$
where $\mathrm{mono}(\sigma)$ is the number of edges in~$E^k$
that are monochromatic under the map~$\sigma$.
Note that $\mathrm{mono}(\sigma)$ is an even number, since $k$ is.
Thus,
$\sum_{\sigma:V\rightarrow\{1,2,3\}}
y^{\mathrm{mono}(\sigma)}$ is 
a positive number which is
at least $1$
if $G$ has a proper $3$ colouring
and is at most 
$3^n y^k$ 
otherwise.
Choosing $$k = \left\lceil
\frac{\log(4 \cdot 3^n)}{\log(1/y)}
\right\rceil,$$
we have $3^n y^k \leq 1/4$, so
a $2$-approximation to $T(G^k;x,y)$ would enable
us to determine whether or not $G$ is $3$-colourable.
\end{proof}

\bibliographystyle{plain}
\bibliography{TutteBib}

\end{document}